\documentclass[a4paper,11pt]{article}
\usepackage[a4paper, total={6in, 8in}]{geometry}

\usepackage{tikz}
\usepackage{lineno}
\usepackage{algorithm} 
\usepackage{algpseudocode} 
\usepackage{xcolor}
\usepackage{graphicx}
\usepackage{a4wide, amsmath, amssymb, amsthm}
\usepackage{amsmath}
\usepackage{enumitem}

\newtheorem{theorem}{Theorem}[section]
\newtheorem{corollary}[theorem]{Corollary}
\newtheorem{lemma}[theorem]{Lemma}

\newtheorem{observation}[theorem]{Observation}
\newtheorem{definition}[theorem]{Definition}

\newcounter{statementcounter}
\newcounter{claimcounter}

\begin{document}

\title{Odd Cycle Transversal in $P_k$-Free Graphs}

\author{ Akramah Faizi	\thanks{Computer Science, Indiana State University, Indiana, USA. Email: {\tt akramahfaizi@gmail.com>}}
 \and 
 Arash Rafiey\thanks{ Computer Science, Indiana State University, Indiana, USA. Email: {\tt arash.rafiey@indstate.edu}. Research supported in part by Bailey Faculty Fellowship} 
 }

\date{}

\numberwithin{equation}{section}
\newcommand{\kcol}{\textsc{$k$-Coloring}}
\newcommand{\minh}{MHOM($H$)\xspace}
\newcommand{\lhomh}{LHOM($H$)\xspace}

\newcommand{\HH}{\textsf{HOM($H$)}\xspace}
\newcommand{\MHH}{\textsf{MHOM($H$)}\xspace}

\newcommand{\MH}{\textsf{MHOM}\xspace}

\newcommand{\kSCSP}{$k$-\textsf{Strict-CSP}\xspace}
\newcommand{\Sig}{\textsc{Sig}\xspace}
\newcommand{\Pol}{\textsf{Pol}}
\newcommand{\Opt}{\textsf{Opt}\xsapace}
\newcommand{\mc}[1]{\mathcal{#1}}
\newcommand{\mb}[1]{\mathbf{#1}}
\newcommand{\abs}[1]{\ensuremath{\left|#1\right|}}

\newcommand{\argmax}{\operatornamewithlimits{argmax}}






\maketitle

\begin{abstract}
The \textsc{Odd Cycle Transversal} (OCT) problem, which asks for a minimum subset of vertices whose removal renders a graph bipartite, is a central problem in algorithmic graph theory. It is known to be NP-complete even on $P_k$-free graphs for $k \ge 6$. Furthermore, assuming the Unique Games Conjecture (UGC), OCT does not admit a constant-factor approximation algorithm on general graphs.

Motivated by these hardness results, we investigate the approximability of OCT on $P_k$-free graphs. We first establish that the problem becomes polynomial-time solvable on specific subclasses of $P_k$-free graphs, most notably $(P_6, C_3)$-free graphs, by exploiting a structural decomposition into rings of bipartite graphs. Leveraging these tractable substructures as a basis, we present a constant-factor approximation algorithm for OCT on general $P_k$-free graphs. We achieve an approximation ratio of $k-2$ when $k$ is odd and $k-3$ when $k$ is even. These results provide the first nontrivial constant-factor approximations for this class dependent on $k$, aligning with the UGC implication that no approximation factor independent of $k$ is likely to exist.

\end{abstract}


\section{Introduction}

The \textsc{Odd Cycle Transversal} (OCT) problem, also known as \textsc{Graph Bipartization}, is a fundamental vertex-deletion problem: given a graph $G=(V,E)$, find a minimum set $S\subseteq V$ such that $G-S$ is bipartite. Equivalently, $S$ must intersect every odd cycle of $G$. OCT is a canonical representative of ``transversal to a hereditary class'' problems and is closely related to finding a largest induced bipartite subgraph.

\noindent
OCT is NP-complete, already as a special case of vertex-deletion problems to hereditary graph classes \cite{Lewis1980}. 
From the approximation viewpoint, the best currently known polynomial-time guarantee on general graphs is an $O(\sqrt{\log n})$-approximation, obtained via approximation algorithms for \textsc{Min-UnCut} / \textsc{Min-2CNF Deletion} \cite{Agarwal2005}; this remains essentially the state of the art for OCT in unrestricted graphs.
Furthermore, under the \textsc{Unique Games Conjecture} (UGC), OCT becomes dramatically harder: it is NP-hard to approximate within any constant factor. This UGC-based inapproximability is commonly attributed to \cite{BansalKhot2009} and is explicitly used as a hardness benchmark in subsequent OCT literature, including kernelization work \cite{KratschWahlstrom2014}.
Taken together, these results make OCT a compelling testbed for investigating whether—and how—structural restrictions on the input graph can restore approximability beyond what is possible in general graphs.

\paragraph{Why $P_k$-free graphs?}
A major line of research in algorithmic graph theory studies hard problems on hereditary classes defined by forbidding induced subgraphs. In particular, forbidding long induced paths (i.e., working in $P_k$-free graphs) often yields strong structure and enables efficient algorithms for problems that are intractable in general.  For OCT, however, this problem exhibits a sharp threshold: the problem is polynomial-time solvable on $P_5$-free graphs \cite{Agrawal2025P5OCT}, while it becomes NP-hard already on $P_6$-free graphs \cite{DabrowskiEtAl2020}. Thus, in contrast to many classical problems, excluding long induced paths does \emph{not} automatically restore tractability for OCT once $k\ge 6$. This leaves open a natural and practically relevant algorithmic question: even when exact optimization is hard on $P_k$-free graphs, can one obtain approximation guarantees whose quality depends only on $k$ (and not on size of the input graph $G$)?

Since OCT is precisely the distance to $2$-colorability, it is natural to place our work in the broader context of coloring graphs with forbidden induced paths.
The class of $P_k$-free graphs is central to the theory of $\chi$-boundedness and often exhibits strong structural constraints; see the survey of Scott and Seymour \cite{ScottSeymourSurvey}.
On the algorithmic side, several sharp thresholds are known: for every fixed $q$, $q$-\textsc{Colorability} is solvable in polynomial time on $P_5$-free graphs \cite{HoangSawadaShuP5Coloring}, and $3$-\textsc{Colorability} is polynomial-time solvable on $P_6$-free graphs \cite{RanderathSchiermeyerP6ThreeColor}.
More generally, a near-complete complexity classification of $k$-\textsc{Coloring} on $P_t$-free graphs, including NP-completeness of $4$-\textsc{Coloring} on $P_7$-free graphs and of $5$-\textsc{Coloring} on $P_6$-free graphs was initiated in \cite{Huang2016PtFreeColoring}.
The remaining borderline case of $4$-\textsc{Coloring} on $P_6$-free graphs was later resolved by a polynomial-time algorithm \cite{ChudnovskySpirklZhong2019P6FourColoring,ChudnovskyEtAlExcellentPrecoloring}.
Together, these results illustrate that forbidding long induced paths can create fine-grained complexity transitions for coloring, and they motivate our study of OCT within the same $P_k$-free hierarchy, where exact solvability breaks at $k=6$ but meaningful approximation guarantees may be achievable.

\paragraph{Our focus and contributions.}
Motivated by (i) the UGC-based barrier against constant-factor approximation on general graphs \cite{BansalKhot2009,KratschWahlstrom2014}, and (ii) the sharp $P_5$ versus $P_6$ tractability threshold \cite{Agrawal2025P5OCT,DabrowskiEtAl2020}, we study the approximability of OCT on $P_k$-free graphs for $k\ge 6$.

\paragraph*{Our Contributions.}
\begin{enumerate}
    \item \textbf{Polynomial-time solvability on structured subclasses.}
    We identify a tractable ``backbone'' subclass of $P_k$-free graphs obtained by additionally excluding all odd cycles up to a length determined by $k$. We show that graphs in this subclass admit a decomposition into a cyclic arrangement of bipartite interactions, and we reduce OCT on these ring-like instances to a sequence of Minimum Vertex Cover computations in bipartite graphs, solvable in polynomial time. The case $k=6$ is more difficult than the case $k>6$, as it requires to established several structural properties for $(P_6,C_3)$-free graphs. In particular we show that OCT problem is polynomial time solvable in $(P_6,C_3)$-free graphs. 

    \item \textbf{$k$-dependent constant-factor approximation on $P_k$-free graphs.}
    Building on the tractable backbone, we give the first polynomial-time approximation guarantees for OCT problem on $P_k$-free graphs whose ratio depends only on $k$. The algorithm greedily packs short odd cycles and removes them, after which the remaining instance falls into the tractable subclass handled exactly. This yields an approximation ratio of $k-2$ for odd $k$ and ratio $k-3$ for even $k$. The running time of our algorithm is a polynomial whose degree is independent of $k$. 

\end{enumerate}

\noindent
These guarantees are qualitatively consistent with the UGC-based hardness for general graphs: while no constant independent of the instance size is expected in full generality, forbidding induced $P_k$ allows approximation factors that scale only with  parameter $k$.

\section{Preliminaries and Notations}
Let $G$ be a graph. We denote the vertex and edge sets of $G$ by $V(G)$ and $E(G)$, respectively. For simplicity, we denote the edge $(x,y)$ in $G$ as $xy$. Unless otherwise stated, we assume all graphs are finite with no multiple edges or self-loops. For an edge $uv$ in $G$, $u$ is a neighbor of $v$, and $v$ is a neighbor of $u$. For a vertex $x \in V(G)$, $N(x)$ denotes the set of neighbors of $x$, and for a subset $X \subset V(G)$, $N(X)$ is the set of vertices in $V(G) \setminus X$ that have at least one neighbor in $X$. For two disjoint subsets $X$ and $Y$ of $V(G)$, we say $X$ is complete to $Y$ if for every $x \in X$ and $y \in Y$, $xy$ is an edge of $G$. We say $X$ is independent of (or anti-complete to) $Y$ if there is no edge of $G$ between $X$ and $Y$. A complete bipartite graphs is also called a biclique. 

For a set $S \subset V(G)$, let $G[S]$ be the subgraph induced by the vertices of $S$; if $xy$ is an edge of $G$ with $x, y \in S$, then $xy \in E(G[S])$. Let $P_k$ denote an induced path on $k$ vertices; it has vertices $v_1, v_2, \dots, v_k$ and edges $v_iv_{i+1}$ for $1 \le i \le k-1$. Similarly, $C_k$ is an induced cycle on $k$ vertices, which has exactly $k$ edges. Let $V_0, V_1, \dots, V_{k-1}$ be disjoint subsets of $V(G)$. We say $V_0, V_1, \dots, V_{k-1}$ form a \emph{$k$-ring} $W$ if every edge of $G[V_0 \cup V_1 \cup \dots \cup V_{k-1}]$ is between $V_i$ and $V_{i+1}$, $0 \le i \le k-1$ (with indices modulo $k$). We say $k$-ring $W$ is complete if $G[V_i \cup V_{i+1}]$ is complete bipartite graph (biclique). Notice that by definition each $V_i$ is an independent set. For two distinct graphs $X$ and $Y$, we say $G$ is $(X, Y)$-free if there is no induced subgraph of $G$ isomorphic to $X$ or $Y$. For example, if $G$ is a $(P_6, C_3)$-free graph, then there is no induced $P_6$ in $G$ and there is no triangle in $G$.
\emph{Chain bipartite graph:} A bipartite graph $H = (X, Y, E)$ is a \emph{chain bipartite graph} (or simply a \emph{chain graph}) if the neighborhoods of the vertices in $X$ can be linearly ordered by inclusion. That is, there exists an ordering of the vertices in $X$, say $x_1, x_2, \dots, x_{|X|}$, such that $N(x_1) \subseteq N(x_2) \subseteq \dots \subseteq N(x_{|X|})$. By symmetry, this condition implies that the neighborhoods of the vertices in $Y$ are also linearly ordered by inclusion. Equivalently, a bipartite graph is a chain graph if and only if it does not contain an induced $2K_2$ (two disjoint edges with no cross-edges between their endpoints).

\section{($P_6,C_3$)-free graphs}
\paragraph{Basic Properties and partition of vertices in $G$:}
Let $G = (V, E)$ be a $(P_6, C_3)$-free graph containing an induced $5$-cycle $C = \{a, b, c, d, e\}$. Let $A_1, B_1, C_1, D_1, E_1$ be the sets of vertices adjacent to exactly one vertex on the $C_5$. Specifically:
\begin{itemize}
    \item $A_1 = \{v \in V(G) \setminus V(C_5) : N(v) \cap V(C) = \{a\}\}$
    \item $B_1, C_1, D_1, E_1$ are defined similarly for vertices $b, c, d, e$ respectively. We often refer to these sets with $X_1$ and $Y_1$. For $X_1 \ne Y_1$ from $\{A_1,B_1,C_1,D_1,E_1\}$, we say $X_1$ and $Y_1$ are consecutive if $(X_1,Y_1)=(A_1,B_1),(B_1,C_1),(C_1,D_1), (D_1,E_1),(E_1,A_1)$ otherwise we say they are not consecutive. 
\end{itemize}
\begin{observation}\label{obs1}
Any vertex $v \in V(G) \setminus V(C)$ is adjacent to at most two vertices on $V(C)$, and these neighbors must be non-consecutive.
\end{observation}
We have the following lemma for adjacencies within and between the sets in $ \{A_1,B_1,C_1,D_1,E_1\}$.
\begin{lemma}\label{AdjacencyX_1}
Let $X_1,Y_1 \in \{A_1,B_1,C_1,D_1,E_1
\}$. Each $X_1$ is independent. Furthermore,  $X_1$ is complete to $Y_1$ if $X_1$ and $Y_1$ are not consecutive, otherwise, $X_1$ and $Y_1$ are independent. 
\end{lemma}
\begin{proof}
Up to symmetry, assume $X_1=A_1$. If $u, v \in A_1$ were adjacent, $\{u, v, a\}$ would form a $C_3$. Suppose $Y_1=B_1$ ($X_1,Y_1$ consecutive). If $a' \in A_1$ and $b' \in B_1$ were adjacent, then $(b', a', a, e, d, c,)$ would form an induced $P_6$. Similarly, $X_1$ and $Y_1$ are independent when $(X_1,Y_1) \in \{(A_1,B_1),(B_1,C_1),(C_1,D_1), (D_1,E_1),(E_1,A_1)\}$.\\

$A_1$ is complete to $C_1$ and $D_1$. If $a' \in A_1$ and $c' \in C_1$ were non-adjacent, then $(a', a, e, d, c, c')$ would form an induced $P_6$. Thus, $X_1$ and $Y_1$ are complete for every $(X_1,Y_1) \in \{(A_1,C_1),(B_1,D_1)\linebreak, (C_1,E_1),(D_1,A_1),(E_1,B_1)\}$.

\end{proof}

\begin{center}
\includegraphics[scale=0.40]{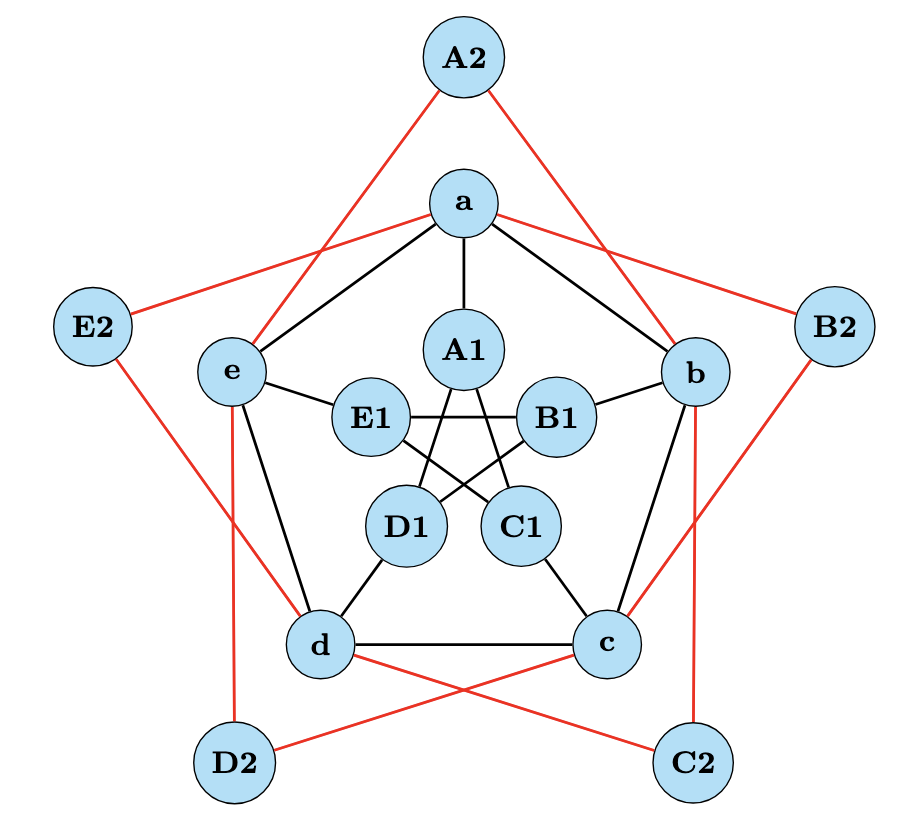}
\end{center}

\begin{definition}[Two neighbors]
    Let $X_2 \subseteq V(G)$ be the set of vertices having exactly two non-consecutive neighbors on the cycle: $A_2$: adjacent to $\{b, e\}$, $B_2$: adjacent to $\{a, c\}$, $C_2$: adjacent to $\{b, d\}$, $D_2$: adjacent to $\{c, e\}$, $E_2$: adjacent to $\{d, a\}$.
\end{definition}

\begin{definition}[Zero neighbor vertices]
   Let $X_0 \subseteq V(G)$ be the set of vertices at distance 2 from $V(C)$. $A_0$ is a set of vertices adjacent to $A_2$, and similarly we define $B_0,C_0,D_0,E_0$. Note that every vertex $x$ of distance two from $C$ is not adjacent to any vertex in $A_1 \cup B_1 \cup C_1 \cup D_1 \cup E_1$, which will be proved later.
\end{definition}

\begin{lemma}[Distance Lemma]\label{lem:distance-2}
When $G$ is connected then every vertex $v \in V(G)$ is at distance at most $2$ from $C$.
\end{lemma}

\begin{proof}
Suppose, for the sake of contradiction, that there exists a vertex in $G$ at a distance of $3$ or more from the cycle $C$.
Let $x$ be a vertex strictly at distance $3$ from $C$. There exists a shortest path $P_{shortest} = (x, y, z, u)$ from $x$ to $C$, where $u \in V(C)$. Thus, $z$ is at distance one from $C$, $y$ is at distance $2$.
Due to Observation \ref{obs1}, we can find a $Q=(u,v,w)$ on $C$ of length two where $z$ is adjacent to $u$ only. Now $P=(x,y,z,u,v,w)$ is an induced $P_6$, a contradiction. 
Therefore, no vertex can be at a distance of $3$ or more from $C$.
\end{proof}

\begin{lemma}[Layer 1 and Layer 2 Dichotomy]\label{X1-X2-dichotomy}
If none of $X_1, X_2, Y_1$ is empty, where $X_1$ and $Y_1$ are not consecutive (i.e. $(X_1,X_2,Y_1)=(A_1,A_2,C_1)$ then  $X_2=X^1_2 \cup X^2_2$ so that $X^2_2$ is complete to $X_1$ and independent from $Y_1$, and $X^1_2$ is complete to $Y_1$ and independent from $X_1$. 
\end{lemma}
\begin{proof}
 Without loss of generality, we map the sets to the specific cycle anchors of $C = (a, b, c, d, e)$. Assume $X_1 = A_1$ (adjacent only to $a$), $Y_1 = C_1$ (adjacent only to $c$), and $X_2 = A_2$ (adjacent exactly to $b$ and $e$). 

Assume moreover that $A_1$ is complete to $C_1$ (as given by Lemma~\ref{AdjacencyX_1}). Define
$
A_2^{2}=\{x\in A_2 : N(x)\cap A_1\neq\emptyset\},\qquad
A_2^{1}=A_2\setminus A_2^{2}.$
Then:
\begin{enumerate}
\item $A_2^{2}$ is complete to $A_1$ and independent from $C_1$;
\item $A_2^{1}$ is independent from $A_1$ and complete to $C_1$.
\end{enumerate}

\smallskip
\noindent\textbf{(1) $A_2^{2}$ is complete to $A_1$ and independent  to $C_1$.}
Let $x\in A_2^{2}$. By definition, $x$ has a neighbor $a_1\in A_1$.
Let $c_1\in C_1$. If $xc_1\in E(G)$, then since $A_1$ is complete to $C_1$ we have $a_1c_1\in E(G)$, and thus $\{a_1,x,c_1\}$ is a triangle, contradicting that $G$ is $C_3$-free.
Hence $x$ has no neighbor in $C_1$, i.e., $A_2^{2}$ is independent from $C_1$.

It remains to show that $x$ is adjacent to every vertex of $A_1$.
Suppose for contradiction that there exists $a_1'\in A_1$ such that $a_1'x\notin E(G)$.
Fix any $c_1\in C_1$ and consider  sequence
\[
P=(c,\ c_1,\ a_1',\ a,\ e,\ x).
\]
The consecutive pairs are edges: $cc_1$ (since $c_1\in C_1$), $c_1a_1'$ (since $A_1$ is complete to $C_1$), $a_1'a$ (definition of $A_1$), $ae$ (edge of $C$), and $ex$ (definition of $A_2$).
Moreover, by the layer definitions relative to $C$ and since $C$ is induced, all nonconsecutive pairs in $P$ are non-edges except possibly $c_1x$.
But $c_1x\notin E(G)$ by the independence just proved.
Therefore $P$ induces a $P_6$, contradicting that $G$ is $P_6$-free.
This contradiction shows that $x$ is adjacent to every $a_1'\in A_1$, i.e., $A_2^{2}$ is complete to $A_1$.

\smallskip
\noindent\textbf{(2) $A_2^{1}$ is independent from $A_1$ and complete to $C_1$.}
Let $y\in A_2^{1}$. By definition of $A_2^{1}$, we have $N(y)\cap A_1=\emptyset$, so $A_2^{1}$ is independent from $A_1$.

Now fix an arbitrary $c_1\in C_1$ and choose any $a_1\in A_1$.
Since $y$ has no neighbors in $A_1$, we have $a_1y\notin E(G)$.
Consider the sequence
$
P=(c,\ c_1,\ a_1,\ a,\ e,\ y).
$
As above, the consecutive pairs are edges, and all nonconsecutive pairs are non-edges except possibly $c_1y$ (by the layer definitions and because $C$ is induced).
If $c_1y\notin E(G)$, then $P$ is an induced $P_6$, contradicting that $G$ is $P_6$-free.
Hence $c_1y\in E(G)$.
Since $c_1\in C_1$ was arbitrary, $y$ is adjacent to every vertex of $C_1$, i.e., $A_2^{1}$ is complete to $C_1$. 
\end{proof}

\vspace{3mm}

We establish a sequence of lemmas that characterize the adjacency between consecutive sets in the second neighborhood, $X_2$ and $Y_2$, based on  their connections to $X_1$ and $Y_1$.

\begin{lemma}\label{A1~A2--B2}
Suppose $X_1,X_2$ and $Y_2$ are not empty where $X_2$ and $Y_2$ are consecutive. Let $X^2_2$ be a subset of $X_2$ with neighbors in $X_1$. Then $X^2_2$ is complete to $Y_2$.

\end{lemma}
\begin{proof}
Up to symmetry, we assume $X_1=A_1$, $X_2=A_2$. Let $A^2_2$ be a subset of $A_2$ where each vertex has a neighbor in $A_1$. Let $a_2 \in A^1_2$, $b_2 \in B_2$, and $a_1 \in A_1$, where $a_2a_1 \in E(G)$. By definition $a_1e \notin E(G)$. Now $P=(a_1, a_2, e, d, c, b_2)$ is an induced $P_6$ unless $a_2b_2 \in E(G)$. 
\end{proof}

\begin{lemma}\label{A1!~A2--B1!~B2--C1}
Suppose $X_1, X_2, Y_1, Y_2, Z_1$ are non-empty where $X_1$ and $X_2$ are independent, $Y_1$ and $Y_2$ are independent, and $Z_1$ is not next to $X_1, Y_1$. Then $X_2$ and $Y_2$ are independent.
\end{lemma}
\begin{proof}
We assume $(X_1, X_2, Y_1, Y_2, Z_1) = (A_1, A_2, B_1, B_2, D_1)$. By Lemma \ref{X1-X2-dichotomy}, since $A_1$ and $A_2$ are independent, $A_2$ and $D_1$ are complete. Similarly, $B_2$ and $D_1$ are complete. Now, for any $a_2 \in A_2$ and $b_2 \in B_2$, they are not adjacent, as otherwise we have a triangle $\{a_2, d_1, b_2\}$ in $G$.
\end{proof}

\begin{lemma}\label{B0!~A1}
Every vertex $x$ of distance two from $C$ is not adjacent to any vertex in $A_1 \cup B_1 \cup C_1 \cup D_1 \cup E_1$.
\end{lemma}
\begin{proof}
Suppose $x$ is adjacent to $a_1 \in A_1$. By definition of distance, $x$ has no neighbors in $V(C)$. Thus, the path $(x,a_1,a,b,c,d)$ is induced. Since $x$ is not adjacent to $a, b, c, d$, this is an induced $P_6$, a contradiction.
\end{proof}

\begin{lemma}\label{2-disjoint-cycle}
Let $G$ be a connected $(P_6, C_3)$-free graph. If $G$ contains two vertex-disjoint induced 5-cycles $C^1$ and $C^2$, then there exists at least one edge with one endpoint in $V(C^1)$ and the other in $V(C^2)$.
\end{lemma} 
\begin{proof}
Assume for contradiction that there is no edge between $V(C^1)$ and $V(C^2)$.
Since $G$ is connected, let $d(C^1,C^2)$ be the minimum distance between a vertex of $C^1$ and a vertex of $C^2$.
Pick any $x\in V(C^2)$. By Lemma~\ref{lem:distance-2}, $d(x,C^1)\le 2$, hence $d(C^1,C^2)\le 2$.
Because we assumed there is no edge between $C^1$ and $C^2$, we have $d(C^1,C^2)\neq 1$, so $d(C^1,C^2)=2$.

Let $u_0\in V(C^1)$ and $v_0\in V(C^2)$ be vertices with $d(u_0,v_0)=2$, and let $w_1$ be a common neighbor on a shortest path
$u_0-w_1-v_0$.
Write the induced 5-cycles as
$C^1=(u_0,u_1,u_2,u_3,u_4,u_0)$ and $C^2=(v_0,v_1,v_2,v_3,v_4,v_0)$.

\smallskip
\noindent\emph{Step 1.} Consider the vertex sequence
$Q=(u_4,u_0,w_1,v_0,v_1,v_2).$
All consecutive pairs are edges by construction. We claim that, among the nonconsecutive pairs in $Q$,
every pair is a non-edge except possibly $w_1v_2$:
\begin{itemize}
\item $u_4w_1\notin E(G)$, otherwise $u_4u_0w_1$ is a triangle, contradicting $C_3$-freeness.
\item $w_1v_1\notin E(G)$, otherwise $w_1v_0v_1$ is a triangle.
\item $u_i v_j\notin E(G)$ for $i\in\{4,0\}$ and $j\in\{0,1,2\}$ by the assumption that there are no edges between $C^1$ and $C^2$.
\item $v_0v_2\notin E(G)$ because $C^2$ is an induced 5-cycle (so it has no chords).
\end{itemize}
Therefore, if $w_1v_2\notin E(G)$, then $Q$ induces a $P_6$, contradicting that $G$ is $P_6$-free.
Hence $w_1v_2\in E(G)$.

\smallskip
\noindent\emph{Step 2.} Now consider
$Q'=(u_4,u_0,w_1,v_2,v_3,v_4).$
Again all consecutive pairs are edges (using $w_1v_2\in E(G)$ from Step 1). As above, all nonconsecutive pairs in $Q'$ are non-edges except possibly $w_1v_4$:
\begin{itemize}
\item $u_4w_1\notin E(G)$ as before.
\item $w_1v_3\notin E(G)$, otherwise $w_1v_2v_3$ is a triangle.
\item $u_i v_j\notin E(G)$ for $i\in\{4,0\}$ and $j\in\{2,3,4\}$ by the assumption of no edges between the cycles.
\item $v_2v_4\notin E(G)$ because $C^2$ is induced.
\end{itemize}
Thus, if $w_1v_4\notin E(G)$ then $Q'$ induces a $P_6$, contradicting $P_6$-freeness.
Hence $w_1v_4\in E(G)$.

Finally, $w_1$ is adjacent to both $v_0$ (by the choice of the shortest path) and $v_4$ (by Step 2), and $v_0v_4\in E(G)$ since they are adjacent on $C^2$.
Therefore $\{w_1,v_0,v_4\}$ spans a triangle, contradicting that $G$ is $C_3$-free.
This contradiction shows that an edge between $V(C^1)$ and $V(C^2)$ must exist.
\end{proof}

\begin{lemma}\label{A1-C2}
Suppose $X_1$ and $Y_2$ are not empty where $X_1,Y_1$ are not consecutive (i.e. $A_1$ and $C_2$ are not empty or $A_1$ and $D_2$ are not empty). Then for every two vertices $x_1,x'_1 \in X_1$, $N(x_1) \cap Y_2 = N(x'_1) \cap Y_2$.

\end{lemma}
\begin{proof}
Up to symmetry we assume $X_1=A_1$ and $Y_2=C_2$. Suppose $a_1c_2$ is an edge where $a_1 \in A_1$ and $c_2 \in C_2$. Now path $P=(a'_1,a,a_1,c_2,d,c)$ is an induced $P_6$ if $a'_1c_2 \not\in E(G)$, $a'_1 \in A_1$. Therefore, $a'_1c_2 \in E(G)$. This implies $N(a_1) \cap C_2= N(a'_1) \cap C_2$. 
\end{proof}

\vspace{3mm}

\vspace{3mm}
The following Corollary follows from Lemma \ref{A1-C2}.
\begin{corollary}\label{A1~C2_2}
    Suppose $X_1$ and $Y_2$ are not empty where $X_1$ and $Y_1$ are not consecutive (e.g. $A_1$ and $C_2$). Then $Y_2= Y^1_2 \cup Y^2_2$ so that $X_1$ and $Y_2^2$ forms a clique and $X_1 \cup Y^1_2$ is independent. 
\end{corollary}

\begin{lemma}\label{A1-D2-C2}
Suppose $X_1,Y_2,Z_2$ (say $A_1,D_2,C_2)$ are not empty where $X_1$ is not consecutive with $Y_1$ and not consecutive with $Z_1$($Y_1,Z_1$ consecutive). Then $Y_2=Y^2_2 \cup Y^1_2$, and $Z_2=Z^1_2 \cup Z^2_2$ so that :
\begin{itemize}
\item $X_1$ is complete to $Y^2_2$ and $Z^2_2$. 
\item $X_1$ is independent of $Y^1_2$ and $Z^1_2$.
\item $Z^1_2$ and $Y^2_2$ form  a complete bipartite graph.
\item $Y^1_2$ and $Z^2_2$ form a complete bipartite graph.

\end{itemize}

\end{lemma}
\begin{proof}
Up to symmetry assume $X_1=A_1$, $Y_2=D_2$ and $Z_2=C_2$. Suppose $c^1_2 \in C^1_2$ is not adjacent to $d^2_2 \in D^2_2$. Now path $P=(c^1_2,d,c,d^2_2,a_1,a)$ is an induced $P_6$ (note that by definition $c^1_2a_1 \not\in E(G)$. Therefore, $D^2_2$ is complete to $C^1_2$. By symmetric argument we conclude that $D^1_2$ is complete to $C^2_2$. Now using Lemma \ref{A1-C2}, we conclude that for every $a_1,a'_1 \in A_1$ either both of them adjacent to all vertices in $C_2$ or neither of them adjacent to $C_2$. Thus, by definition $A_1$ is complete to $C^2_2$. Similarly, $A_1$ is complete to $D^2_2$. \end{proof}

\vspace{3mm}

\begin{lemma}\label{A0-A2-B2}
If $X_0$, $X_2$, and $Y_2$ ($X_2,Y_2$ are consecutive) are not empty then the following hold. 
\begin{itemize}
    \item [(1)] For every two vertices $a_0,a'_0 \in X_0$ with a common neighbor in $X_2$, $N(a_0) \cap Y_2 = N(a'_0) \cap Y_2$. 
    \item [(2)] For every two vertices $a_0,a'_0 \in X_0$ with a common neighbor in $Y_2$, $N(a_0) \cap X_2 = N(a'_0) \cap X_2$. 
    \item [(3)] For every two vertices $a_2,a'_2 \in X_2$ with a common neighbor in $X_0$, $N(a_2) \cap Y_2 = N(a'_2) \cap Y_2$. 
   
\end{itemize}
\end{lemma}

\begin{proof}
Up to symmetry, let $X_0=A_0$, $X_1=A_1$ and $Y_2=B_2$. Suppose $a_0a_2,a'_0a_2$ are edges of $G$ with $a_0,a'_0 \in A_0$ and $a_2 \in A_2$. Let $b_2 \in B_2$ so that $a_0b_2 \in E(G)$. 
Now $P=(a'_0,a_2,a_0,b_2,c,d)$ is an induced $P_6$ unless $a'_0b_2 \in E(G)$. This proves (1). The proof of (2) is analogues. 

Let $a_2,a'_2 \in A_2$ have a common neighbor $a_0 \in A_0$. Let $a_2b_2 \in E(G)$ where $b_2 \in B_2$. Now $P=(a'_2,a_0,a_2,b_2,c,d)$ is an induced $P_6$ unless $a'_2b_2 \in E(G)$, implying $N(a_2) \cap B_2 = N(a'_2) \cap B_2$. This proves (3). 
\end{proof}

\section{Structural properties of 5-Ring in $(P_6,C_3)$-free graphs }
In this section we develop some structural properties for $5$-ring graph.  
Let $H$ be a connected $(P_6,C_3)$-free graph consisting of a ring of 5 vertex sets $V_1, \dots, V_5$. In the rest of this manuscript when referring to these five sets, $V_6=V_1$, and $V_0=V_5$. Since in $(P_6,C_3)$-free graphs every odd cycle has length $5$, we focus on $C_5$ witness, and hence, we assume every vertex of $H$ lies on a $5$-cycle.  Let $H_j = H[V_j \cup V_{j+1}]$ denote the bipartite subgraph induced by consecutive sets.

\begin{lemma}\label{chain-switch}
Let $a_1, a_2 \in V_j$ and $b_1, b_2 \in V_{j+1}$. Suppose that $a_1b_1 \in E(H)$, $a_2b_2 \in E(H)$, and $b_1a_2 \in E(H)$, but $a_1b_2 \notin E(H)$ (a `Z' shape). Then:
\begin{enumerate}
    \item $N(b_2) \cap V_{j+2} \subseteq N(b_1) \cap V_{j+2}$,
    \item $N(a_1) \cap V_{j-1} \subseteq N(a_2) \cap V_{j-1}$.
\end{enumerate}
\end{lemma}

\begin{proof}
We prove each inclusion by contradiction.

\smallskip
\noindent\textbf{Part 1.}
Assume there exists $c\in V_{j+2}$ such that $cb_2\in E(H)$ but $cb_1\notin E(H)$.
Since every vertex lies on an induced $C_5$ in the ring, $c$ has a neighbor $d\in V_{j+3}$.
Consider the sequence $a_1,b_1,a_2,b_2,c,d$.
All consecutive pairs are edges: $a_1b_1$, $b_1a_2$, $a_2b_2$ (hypothesis), $b_2c$ (choice of $c$), and $cd$ (choice of $d$).
We claim this is an induced $P_6$.
Indeed, $a_1b_2\notin E(H)$ by hypothesis and $b_1c\notin E(H)$ by the choice of $c$.
All other nonconsecutive pairs are non-edges because edges occur only between consecutive sets:
$V_j$ is independent to $V_{j+2}\cup V_{j+3}$ and $V_{j+1}$ is independent to $V_{j+3}$.
Also $b_1b_2\notin E(H)$ since $V_{j+1}$ is independent.
Hence we obtain an induced $P_6$, contradicting that $H$ is $P_6$-free.
Therefore no such $c$ exists and $N(b_2)\cap V_{j+2}\subseteq N(b_1)\cap V_{j+2}$.

\smallskip
\noindent\textbf{Part 2.}
Assume there exists $x\in V_{j-1}$ such that $xa_1\in E(H)$ but $xa_2\notin E(H)$.
Since every vertex lies on an induced $C_5$, $x$ has a neighbor $y\in V_{j-2}$.
Consider the sequence $y,x,a_1,b_1,a_2,b_2$.
Again all consecutive pairs are edges: $yx$ and $xa_1$ by choice of $y,x$, and $a_1b_1$, $b_1a_2$, $a_2b_2$ by hypothesis.
We claim this is an induced $P_6$.
We have $xa_2\notin E(H)$ by assumption and $a_1b_2\notin E(H)$ by hypothesis.
Moreover $x$ is independent to $V_{j+1}$, so $xb_1,xb_2\notin E(H)$; and $y$ is independent to $V_j\cup V_{j+1}$, so
$ya_1,ya_2,yb_1,yb_2\notin E(H)$.
Finally $a_1a_2\notin E(H)$ and $b_1b_2\notin E(H)$ since each set is an independent set.
Thus the sequence induces a $P_6$, contradicting $P_6$-freeness.
Therefore no such $x$ exists and $N(a_1)\cap V_{j-1}\subseteq N(a_2)\cap V_{j-1}$.
\end{proof}

\begin{lemma}\label{2K2}
Let $a_1, a_2 \in V_j$ and $b_1, b_2 \in V_{j+1}$. Suppose $a_1b_1 \in E(H)$ and $a_2b_2 \in E(H)$ form an induced $2K_2$
(i.e., $a_1b_2 \notin E(H)$ and $a_2b_1 \notin E(H)$).
Then
\[
N(a_1)\cap V_{j-1}=N(a_2)\cap V_{j-1}
\qquad\text{and}\qquad
N(b_1)\cap V_{j+2}=N(b_2)\cap V_{j+2}.
\]
\end{lemma}

\begin{proof}
We prove the first equality; the second is symmetric (swap the roles of $(V_j,a_1,a_2)$ and $(V_{j+1},b_1,b_2)$ and shift indices).

Assume for contradiction that $N(a_1)\cap V_{j-1}\neq N(a_2)\cap V_{j-1}$.
By symmetry we may assume there exists $e_1\in V_{j-1}$ with $e_1a_1\in E(H)$ and $e_1a_2\notin E(H)$.
Since every vertex lies on an induced $C_5$ and edges exist only between consecutive sets, $a_2$ has some neighbor $e_2\in V_{j-1}$.
Necessarily $e_2\neq e_1$ (because $e_1a_2\notin E(H)$). We first show:

\smallskip
\noindent\textbf{Claim 1.} $e_2a_1\notin E(H)$.

\noindent Indeed, suppose $e_2a_1\in E(H)$. Since $b_2$ lies on a $C_5$, it has a neighbor $c_2\in V_{j+2}$.
Then the vertex sequence
$
P=(e_1,\ a_1,\ e_2,\ a_2,\ b_2,\ c_2)
$
is a $P_6$: all consecutive pairs are edges, and it is induced because
$e_1a_2\notin E(H)$ by choice of $e_1$, $a_1b_2\notin E(H)$ by the $2K_2$ hypothesis, $e_1e_2\notin E(H)$ since $V_{j-1}$ is independent,
and all remaining nonconsecutive pairs are non-edges by the ring distance constraints (edges only between consecutive sets).
This contradicts that $H$ is $P_6$-free. Hence $e_2a_1\notin E(H)$, proving the claim.

Thus, the two edges $e_1a_1$ and $e_2a_2$ form an induced $2K_2$ in $H[V_{j-1}\cup V_j]$.
Next, since every vertex lies on a $C_5$, vertex $b_2$ has a neighbor $c_2\in V_{j+2}$, and vertex $b_1$ has a neighbor $c_1\in V_{j+2}$.
We show these can be chosen so that they are \emph{private} across the $2K_2$:

\smallskip
\noindent\textbf{Claim 2.} $b_1c_2\notin E(H)$ and $b_2c_1\notin E(H)$.

\noindent We prove $b_1c_2\notin E(H)$; the other is symmetric.
If $b_1c_2\in E(H)$, consider the sequence
\[
Q=(a_1,\ b_1,\ c_2,\ b_2,\ a_2,\ e_2).
\]
Consecutive pairs are edges: $a_1b_1$ and $a_2b_2$ by hypothesis, $b_2c_2$ by definition of $c_2$, $e_2a_2$ by definition of $e_2$,
and $b_1c_2$ by assumption. It is induced because
$a_1b_2\notin E(H)$ and $a_2b_1\notin E(H)$ by the $2K_2$ hypothesis, $a_1e_2\notin E(H)$ by Claim~1, and all other nonconsecutive pairs
are forbidden by ring distance constraints (e.g., $V_j$ is independent to $V_{j+2}$).
Thus $Q$ is an induced $P_6$, contradicting $P_6$-freeness. Hence $b_1c_2\notin E(H)$. This proves Claim~2.

\smallskip
Now extend $c_1,c_2$ one more step along the ring.
Since every vertex lies on a $C_5$, each $c_t$ has a neighbor $d_t\in V_{j+3}$ ($t\in\{1,2\}$).
Because $V_{j+3}=V_{j-2}$ in a $5$-ring, each $d_t$ is adjacent to vertices in $V_{j-1}$; moreover we may choose $d_1\in N(e_1)\cap V_{j+3}$ and
$d_2\in N(e_2)\cap V_{j+3}$ (each $e_t$ lies on a $C_5$, hence has a neighbor in $V_{j-2}=V_{j+3}$).
Consider the two cycles
\[
C_1=(e_1,\ a_1,\ b_1,\ c_1,\ d_1,\ e_1)
\qquad\text{and}\qquad
C_2=(e_2,\ a_2,\ b_2,\ c_2,\ d_2,\ e_2).
\]
Each is a $5$-cycle in $H$ because edges exist only between consecutive sets.
Moreover, $C_1$ and $C_2$ are vertex-disjoint (since $e_1\neq e_2$, $a_1\neq a_2$, $b_1\neq b_2$, and Claim~2 ensures $c_1\neq c_2$ unless it would create a common neighbor).
Finally, by construction and Claim~2 there are no edges between $V(C_1)$ and $V(C_2)$ other than possibly edges within the same set;
but each set is independent, so no such cross-edges exist. Hence $C_1$ and $C_2$ are two vertex-disjoint induced $C_5$'s with no edge between them.

This contradicts Lemma~\ref{2-disjoint-cycle} (in a connected $(P_6,C_3)$-free graph, two vertex-disjoint induced $5$-cycles must have an edge between them).
Therefore our assumption was false, and $N(a_1)\cap V_{j-1}=N(a_2)\cap V_{j-1}$.
The equality $N(b_1)\cap V_{j+2}=N(b_2)\cap V_{j+2}$ follows by symmetry.
\end{proof}

We define the direction of a proper chain graph $H_j$ (induced by $V_j \cup V_{j+1}$) based on the inclusion ordering.
\begin{itemize}
    \item $L \to R$: $N(a_{i+1}) \cap V_{j+1} \subseteq N(a_i) \cap V_{j+1}$ (Nested neighborhoods in $V_{j+1}$).
    \item $R \to L$: $N(b_{i+1}) \cap V_j \subseteq N(b_i) \cap V_j$ (Nested neighborhoods in $V_j$).
\end{itemize}

\begin{lemma}\label{lem:alternate-chain}
Suppose $H_j$, $H_{j+1}$, and $H_{j+2}$ are proper chain graphs with alternating inclusion directions ($L \to R$, $R \to L$, and $L \to R$ respectively). Then the adjacent bipartite graphs $H_{j-1}$ and $H_{j+3}$ must be complete bipartite graphs.
\end{lemma}
\begin{proof}
We prove the statement for $H_{j+3}=H[V_{j+3}\cup V_{j+4}]$; the argument for $H_{j-1}$ is symmetric.
Assume for a contradiction that $H_{j+3}$ is not complete bipartite. Then there exist vertices
$d_1\in V_{j+3}$ and $e_1\in V_{j+4}$ such that $d_1e_1\notin E(H)$.

Because $H_{j+2}=H[V_{j+2}\cup V_{j+3}]$ is a \emph{proper} chain graph and the inclusion direction is $L\to R$,
there exist $c_1,c_2\in V_{j+2}$ and $d_1,d_2\in V_{j+3}$ forming a $Z$-configuration, and hence, we have 
$c_1d_1,\ c_2d_2,\ c_2d_1\in E(H)$
and
$c_1d_2\notin E(H).$
(Equivalently, $N(c_1) \cap V_{j+3}\subsetneq N(c_2) \cap V_{j+3}$.)

Since $H_{j+1}=H[V_{j+1}\cup V_{j+2}]$ is a proper chain graph and the direction alternates ($R\to L$),
we can choose $b_1,b_2\in V_{j+1}$ witnessing strict inclusion in the opposite direction so that
$b_1c_1,\ b_1c_2,\ b_2c_2\in E(H)$
and
$b_2c_1\notin E(H),$
i.e., $N(b_2) \cap V_{j+2} \subsetneq N(b_1) \cap {V_{j+2}}$.

Similarly, since $H_j=H[V_j\cup V_{j+1}]$ is proper and the direction is again $L\to R$,
there exist $a_1,a_2\in V_j$ such that
$
a_1b_1,\ a_2b_2,\ a_2b_1\in E(H)$
and
$a_1b_2\notin E(H),
$
i.e., $N(a_1) \cap V_{j+1} \subsetneq N(a_2) \cap V_{j+1}$.

\smallskip
\noindent\textbf{Step 1: propagate inclusion to $V_{j+4}$.}
Apply Lemma~\ref{chain-switch} to the $Z$-shape in $H_{j+2}$ (with $a_1=c_1$, $a_2=c_2$, $b_1=d_1$, $b_2=d_2$).
We obtain $N(d_2)\cap V_{j+4}\subseteq N(d_1)\cap V_{j+4}.$

Since every vertex lies on a $5$-cycle in the ring, $d_2$ has a neighbor $e_2\in V_{j+4}$.
The above inclusion then implies $d_1e_2\in E(H)$.

\smallskip
\noindent\textbf{Step 2: propagate inclusion to $V_{j-1}=V_{j+4}$.}
Apply Lemma~\ref{chain-switch} to the $Z$-shape in $H_j$ (with $a_1=a_1$, $a_2=a_2$, $b_1=b_1$, $b_2=b_2$).
We obtain
$N(a_1)\cap V_{j-1}\subseteq N(a_2)\cap V_{j-1}.$

Since every vertex lies on a $5$-cycle, $a_1$ has a neighbor in $V_{j-1}$; choose one and denote it by $e_1\in V_{j-1}(=V_{j+4})$.
Then the above inclusion implies $a_2e_1\in E(H)$.

\smallskip
\noindent\textbf{Step 3: force a chord to avoid an induced $P_6$.}
Consider the vertex sequence
\[
P=(d_2,\ e_2,\ d_1,\ c_1,\ b_1,\ a_1).
\]
Consecutive pairs are edges: $d_2e_2$ by choice of $e_2$, $e_2d_1$ from Step~1, $d_1c_1$ by the $Z$-shape in $H_{j+2}$,
$c_1b_1$ by the $Z$-shape in $H_{j+1}$, and $b_1a_1$ by the $Z$-shape in $H_j$.

We check that, except possibly for $e_2a_1$, there are no chords among these six vertices.
Indeed, by the ring property edges occur only between consecutive bags; hence:
$d_2\in V_{j+3}$ is anticomplete to $V_{j+1}\cup V_j$ so $d_2b_1,d_2a_1\notin E(H)$,
$e_2\in V_{j+4}$ is anticomplete to $V_{j+2}\cup V_{j+1}$ so $e_2c_1,e_2b_1\notin E(H)$,
and $c_1\in V_{j+2}$ is anticomplete to $V_j$ so $c_1a_1\notin E(H)$.
Also $d_2c_1\notin E(H)$ since both lie in nonconsecutive bags, and $V_{j+4}$ is independent so $e_1e_2\notin E(H)$.
Therefore, if $e_2a_1\notin E(H)$, the sequence $P$ induces a $P_6$, contradicting that $H$ is $P_6$-free.
Hence we must have $e_2a_1\in E(H)$.

\smallskip
\noindent\textbf{Step 4: conclude $H_{j+3}$ is complete bipartite.}
We have shown that for an arbitrary neighbor $e_2$ of $d_2$ in $V_{j+4}$, we also get $e_2d_1\in E(H)$ (Step~1) and $e_2a_1\in E(H)$ (Step~3).
In particular, the neighborhood inclusion forced by the chain triple prevents missing edges across $V_{j+3}\cup V_{j+4}$:
if some non-edge existed in $H_{j+3}$, then repeating the above construction yields an induced $P_6$.
Therefore $H_{j+3}$ must be complete bipartite.

By symmetry (shifting indices and reversing directions), the same argument shows $H_{j-1}$ is complete bipartite.
\end{proof}

\begin{corollary}\label{disconnected}
Suppose $H_j$ is not connected. Then $H_{j-1}$ and $H_{j+1}$ are both complete bipartite graphs.
\end{corollary}
\begin{proof}
Since $H_j$ is not connected and every vertex lies on a 5-cycle (ensuring no isolated vertices), $H_j$ must contain at least two disjoint edges $a_1b_1$ and $a_2b_2$ (where $a_1, a_2 \in V_j$ and $b_1, b_2 \in V_{j+1}$) such that there are no edges between $\{a_1, b_1\}$ and $\{a_2, b_2\}$. This forms an induced $2K_2$.

By Lemma \ref{2K2}, we must have $N(a_1) \cap V_{j-1} = N(a_2) \cap V_{j-1}$. Since $a_1$ and $a_2$ can be chosen from any two distinct components of $H_j$, it follows that all vertices in $V_j$ must share the same neighborhood in $V_{j-1}$.
Suppose $H_{j-1}$ is not complete. Then there exists $x \in V_{j-1}$ and $y \in V_j$ such that $xy \notin E(H)$. Since all vertices in $V_j$ share the same neighborhood, $x$ is not adjacent to \textit{any} vertex in $V_j$. This contradicts the assumption that every vertex lies on a 5-cycle (as $x$ would be isolated from $V_j$, preventing cycle formation through that link).
Therefore, $H_{j-1}$ must be complete. By symmetry, $H_{j+1}$ must also be complete.
\qed
\end{proof}

\begin{lemma}\label{2K2-structure}
Suppose $H_j$ is connected and contains two independent edges $a_1b_1$ and $a_2b_2$. Then $H_j$ is partitioned into $A_1, A_2, A_3$ where $A_1$ and $A_2$ are connected and contain $a_1b_1$ and $a_2b_2$ (respectively), and $A_1$ and $A_2$ are independent. $A_3$ is complete to both $A_1$ and $A_2$. Furthermore, $N(A_1) \cap V_{j-1} = N(A_2) \cap V_{j-1} \subseteq N(A_3) \cap V_{j-1}$ and $N(A_1) \cap V_{j+2} = N(A_2) \cap V_{j+2} \subseteq N(A_3) \cap V_{j+2}$.

\end{lemma}
\begin{proof}
Let $A_1$ and $A_2$ be two maximal connected subgraphs in $H_j$ containing $a_1b_1$ and $a_2b_2$ respectively such that there are no edges between $A_1$ and  $A_2$. By Lemma \ref{2K2}, $N(A_1) \cap V_{j-1} = N(A_2) \cap V_{j-1}$.
Since $H_j$ is connected, there must be a vertex $a_3 \in V_j$ (without loss of generality) such that $b_1a_3, b_2a_3 \in E(H_j)$.
Suppose $e_3a_3 \in E(H_{j-1})$. Consider the path $P=(e_3, a_3, b_2, a_2, e, a_1)$ where $ea_1 \in E(H_{j-1})$. Since $P$ is not an induced $P_6$, one of $e_3a_1, e_3a_2, e_1a_3$ must be in $E(H_{j-1})$.
Applying Lemma \ref{2K2} to the edges $a_3b_1, a_3b_2$ (which share $a_3$) vs $a_2b_2$, and knowing $a_2b_1 \notin E$, we conclude that $N(a_2) \cap V_{j-1} \subseteq N(a_3) \cap V_{j-1}$.
Therefore, $N(A_1) \cap V_{j-1} = N(A_2) \cap V_{j-1} \subseteq N(A_3) \cap V_{j-1}$. Analogously, we conclude $N(A_1) \cap V_{j+2} = N(A_2) \cap V_{j+2} \subseteq N(A_3) \cap V_{j+2}$.

Now consider a neighbor of $a_1$ in $A_1$, say $b'_1$. The sequence $(b'_1, a_1, b_1, a_3, b_2, a_2)$ forms an induced $P_6$ unless $b'_1a_3 \in E(H)$. Hence, $A_1$ and $A_3$ are complete. Similarly, $A_2$ and $A_3$ are complete.\qed
\end{proof}

\begin{lemma}\label{OCT-chain-chain-chain}
Let $H_j, H_{j+1}, H_{j+2}$ be chain graphs. Suppose $X$ is a set of vertices from $V_{j} \cup V_{j+1} \cup V_{j+2}$ such that $H \setminus X$ is bipartite, i.e. $X$ is an OCT. Then $X$ is a vertex cover for $H[V_{j} \cup V_{j+1}]$ or $X$ is a vertex cover for $H[V_{j+1} \cup V_{j+2}]$.

\end{lemma}
\begin{proof}
By Lemma \ref{chain-switch}, let $a_1, \dots, a_{l_1}$ be the vertices in $V_j$ ordered such that $N(a_{i+1}) \cap V_{j+1} \subseteq N(a_i) \cap V_{j+1}$ ($L \to R$).
Let $b_1, \dots, b_{l_2}$ be vertices in $V_{j+1}$ ordered such that $N(b_{i}) \cap V_{j+2} \subseteq N(b_{i+1}) \cap V_{j+2}$ ($R \to L$).
Finally, let $c_1, \dots, c_{l_3}$ be vertices in $V_{j+2}$ ordered such that $N(c_{i+1}) \cap V_{j+1} \subseteq N(c_i) \cap V_{j+1}$ ($L \to R$).

If $a_i \notin X$ while $a_{i+1} \in X$, then due to the inclusion property, we can swap $a_{i+1}$ out of $X$ and put $a_i$ into $X$ without losing coverage. Thus, we may assume $X$ contains prefixes/suffixes: $a_1, \dots, a_p \in X$; $b_t, \dots, b_{l_2} \in X$; and $c_1, \dots, c_r \in X$.
This means $c_{r+1}, b_{t-1}, b_{t-2}, a_{p+1} \notin X$.

If there is no edge $b_{t'}c_{r'}$ in $H_{j+1}$ with $t' < t$ and $r < r'$, then $X \cap H_{j+1}$ is a vertex cover for $H_{j+1}$, proving the lemma.
Thus, assume $b_{t-1}c_{r+1} \in E(H_{j+1})$.
Now consider $H_j$. If $X$ is not a vertex cover for $H_j$, there must be an edge $a_{p+1}b_{k}$ surviving, where $k < t$. 
Due to inclusion in $V_{j+1}$, $N(b_{t-2}) \cap V_j \subseteq N(b_{t-1}) \cap V_j$. Thus if $a_{p+1}$ connects to any $b_k (k<t)$, it must connect to $b_{t-1}$.
Therefore, $a_{p+1}b_{t-1} \in E(H)$.
Since $b_{t-1}c_{r+1} \in E(H)$ and $a_{p+1}b_{t-1} \in E(H)$, we have a surviving path of length 2. This creates a conflict with the global odd cycle constraint, contradicting $X$ being an OCT. \qed 
\end{proof}

\begin{lemma}\label{lem:OCT-2K2}
Let $H_{j} = H[V_j \cup V_{j+1}]$ be connected and contain an induced $2K_2$. Suppose $X$ is a minimum OCT for $G$ containing vertices from $V_{j-1}, V_{j}, V_{j+1}$, and $V_{j+2}$. Then $X$ must contain a vertex cover for $H_{j-1}$, $H_j$, or $H_{j+1}$.
\end{lemma}
\begin{proof}
Since $G$ is structured as an odd ring of bipartite graphs, any valid OCT $X$ must act as a vertex cover for at least one bipartite subgraph $H_i$ to successfully break all macroscopic odd cycles.

By Lemma \ref{lem:gen-2K2}, the neighborhoods of the $2K_2$ vertices in $V_{j-1}$ are identical. Let this shared neighborhood be $Y \subseteq V_{j-1}$. 
According to Lemma \ref{2K2-structure}, because $H_j$ is connected and contains this $2K_2$, the vertices of $V_j$ can be partitioned into sets $A_1, A_2, A_3$ such that $Y$ is completely connected to the entire active component of $V_j$ (specifically $A_1 \cup A_2 \cup A_3$). Thus, the subgraph induced by $Y \cup V_j$ contains a complete bipartite graph (a biclique) between $Y$ and $V_j$.

Similarly, let $Z \subseteq V_{j+2}$ be the shared neighborhood of the $2K_2$ vertices in $V_{j+1}$. By symmetry, $Z$ is completely connected to the active component of $V_{j+1}$, forming a biclique between $V_{j+1}$ and $Z$.

To be a valid OCT, $X$ must cover all edges in these dense bicliques. A fundamental property of a minimum vertex cover on a complete bipartite graph $(U, W)$ is that it must completely contain either $U$ or $W$; any partial selection from both sides is strictly sub-optimal or fails to cover the biclique.

We evaluate the structural requirements for $X$:
\begin{enumerate}
    \item \textbf{Covering the Left Biclique $(Y, V_j)$:} 
    $X$ must either fully contain $Y$ or fully contain $V_j$. 
    If $X$ fully contains $Y$, it covers all active edges in $H_{j-1}$. Thus, $X$ contains a vertex cover for $H_{j-1}$.
    
    \item \textbf{Covering the Right Biclique $(V_{j+1}, Z)$:} 
    $X$ must either fully contain $Z$ or fully contain $V_{j+1}$.
    If $X$ fully contains $Z$, it covers all active edges in $H_{j+1}$. Thus, $X$ contains a vertex cover for $H_{j+1}$.
    
    \item \textbf{Covering the Center $H_j$:}
    Suppose $X$ contains neither a vertex cover for $H_{j-1}$ (meaning $Y \not\subseteq X$) nor a vertex cover for $H_{j+1}$ (meaning $Z \not\subseteq X$). 
    To cover the left biclique, $X$ is forced to fully contain $V_j$. To cover the right biclique, $X$ is forced to fully contain $V_{j+1}$. 
    If $X$ contains all active vertices of $V_j$, it covers all edges of $H_j$. Therefore, $X$ serves as a full vertex cover for $H_j$.
\end{enumerate}

Consequently, the dense biclique structures flanking the $2K_2$ force any optimal OCT $X$ to make an ``all-or-nothing'' choice, guaranteeing that $X$ completely covers at least one of the subgraphs $H_{j-1}$, $H_j$, or $H_{j+1}$. \qed

\end{proof}

\begin{lemma}\label{OCT-independent}
Let $H_{j}$ have more than one connected component. Suppose $X$ is a minimum OCT containing vertices from $V_{j-1}, V_{j}, V_{j+1}, V_{j+2}$. Then $X$ must contain a vertex cover in $H_{j-1}$ or in $H_j$ or in $H_{j+1}$.
\end{lemma}

\begin{proof}
By Corollary \ref{disconnected}, since $H_{j-1}$ is a complete bipartite graph, if $X$ contains vertices in $V_{j-1}$, it must contain all of $V_{j-1}$, making $X$ a vertex cover for $H_{j-1}$. Similarly, if $X$ contains vertices from $V_{j+2}$, then $X=V_{j+2}$, making it a vertex cover for $H_{j+1}$.
Thus, we may assume $X \cap V_{j+2} = X \cap V_{j-1} = \emptyset$. Thus, $X$ is inside $H_{j}$, and hence, it must be a vertex cover in $H_j$.
\end{proof}

\begin{theorem}[Bottleneck Cut]\label{tm:ring-cut}
Let $H$ be a $(P_6, C_3)$-free graph consisting of a ring of 5 vertex sets $V_1, \dots, V_5$ where every vertex lies on a 5-cycle.
The size of the Minimum Odd Cycle Transversal of $H$ is equal to the minimum size of a Vertex Cover among the five bipartite subgraphs induced by consecutive sets $(V_i, V_{i+1})$.
\[ OCT(H) = \min_{i=1}^5 \{ \tau(H[V_i \cup V_{i+1}]) \} \]
where $\tau(G)$ denotes the vertex cover number.
\end{theorem}

\begin{proof}
Let $S$ be a minimum vertex cover of $H[V_i \cup V_{i+1}]$. Removing $S$ eliminates all edges between $V_i$ and $V_{i+1}$. Since $H$ forms a ring, any odd cycle in $H$ must traverse the ring sequence $V_1 \to \dots \to V_5 \to V_1$, removing this link breaks all such cycles. Thus $S$ is an OCT.
Conversely, according to Lemmas \ref{OCT-chain-chain-chain},\ref{lem:OCT-2K2}, and \ref{OCT-independent}, any minimum OCT $X$ must be a vertex cover for at least one $H_j$.
\end{proof}



\section{Cycle decomposition and algorithm for OCT  in $(P_6,C_3)$-free graphs}

Based on the structural decomposition derived previously, the vertex set $V$ is partitioned into a finite collection of independent sets: 
\[ \mathcal{S} = \{A_1, B_1, \dots, E_1, A_2, B_2, \dots, E_2, A_0, \dots, E_0\} \]
where each set in $\mathcal{S}$ represents a specific module or layer relative to the cycle $C$.

We assume there are some 5-cycles in $G_{-5}=G \setminus V(C)$. Otherwise, we check each subset of $V(C)$ to be removed so that the remaining graph becomes bipartite. We also notice that every odd cycle in $(P_6,C_3)$-free graph has length $5$, thus when it comes to the optimal OCT, we assume every vertex in $G_{-5}$ lies on a 5-cycle; otherwise, that vertex can be removed and it won't be part of an optimal OCT solution.

   

\subsection{Algorithm Description}\label{sec:alg}

\noindent\textbf{Input:} A connected $(P_6,C_3)$-free graph $G$.\\
\textbf{Output:} A minimum odd cycle transversal $S\subseteq V(G)$.

\medskip
\noindent\textbf{High-level idea.}
By $P_6$-freeness, every induced odd cycle in $G$ has length $5$ (and $C_3$-freeness excludes triangles). 
We therefore search for a constant-sized \emph{ring witness} inside the coarse partition, refine it into a chordless $5$-ring subgraph $H$, and solve OCT on $H$ using the Bottleneck Cut Theorem (Theorem~\ref{tm:ring-cut}). We then recurse on the remaining graph.

\medskip
\noindent\textbf{Step 0: find a base $C_5$ and build the coarse partition.}
If $G$ is bipartite, return $\emptyset$.
Otherwise find an induced $5$-cycle $C=(a,b,c,d,e)$ (e.g. by BFS from each vertex until an odd cycle is found, then chord-minimize).
Construct the sets $X_1,X_2,X_0$ for $X\in\{A,B,C,D,E\}$ as in Section~3, and let
\[
\mathcal S=\{A_1,B_1,C_1,D_1,E_1,\ A_2,B_2,C_2,D_2,E_2,\ A_0,B_0,C_0,D_0,E_0\}.
\]
Build the coarse graph $G_{\mathcal S}$ whose vertices are the nonempty sets in $\mathcal S$ and where $UW\in E(G_{\mathcal S})$ iff there exists at least one edge between $U$ and $W$ in $G$.

\medskip
\noindent\textbf{Step 1: enumerate coarse $5$-cycles and extract  $5$-rings.}
Enumerate all $5$-cycles $W=(\Pi_1,\Pi_2,\Pi_3,\Pi_4,\Pi_5,\Pi_1)$ in the constant-size graph $G_{\mathcal S}$.
For each such $W$, we construct a $5$-ring instance as follows.

\smallskip
\noindent\emph{Refinement subroutine \textsc{ExtractRing}$(W)$.}
Initialize $\Delta_i:=\Pi_i$ for $i=1,\dots,5$.
If there exists an edge in $G$ between some nonconsecutive pair $\Delta_i$ and $\Delta_{i+2}$ (indices mod $5$), refine the involved layer-2 set(s) using Lemma~3.6 (and its rotated versions) to split $X_2$ into $X_2^1\cup X_2^2$ so that nonconsecutive adjacencies become uniform.
Replace $\Delta_i$ by the appropriate refined part(s) and discard any refined parts that do not participate in any $5$-cycle through the five partitions.
Since each layer-2 set splits into at most two parts, and $|V(G_{\mathcal S})|\le 15$, this refinement produces only constantly many candidate rings per coarse $5$-cycle.

At the end of refinement, we obtain a \emph{chordless} $5$-ring
\[
W_r=(\Delta_1,\Delta_2,\Delta_3,\Delta_4,\Delta_5,\Delta_1),
\]
meaning that every edge of $G[\Delta_1\cup\cdots\cup\Delta_5]$ is between consecutive sets.

\smallskip
\noindent\emph{Ring-core cleanup.}
Let $H$ be the induced subgraph $G[\Delta_1\cup\cdots\cup\Delta_5]$.
We restrict $H$ to the vertices that lie on a (necessarily induced) $5$-cycle within $H$:
compute the set $U\subseteq V(H)$ of vertices that appear in some $C_5$ of $H$,
and replace $H$ by $H[U]$ and $\Delta_i$ by $\Delta_i\cap U$.
(Vertices outside $U$ cannot lie on any odd cycle of $H$, so they are irrelevant for OCT on $H$.)

Return the cleaned ring instance $H$ together with its sets $(\Delta_1,\dots,\Delta_5)$.

\medskip
\noindent\textbf{Step 2 (Branch on the selected edge).}
Let $H=G[\Delta_1\cup\cdots\cup\Delta_5]$ be the induced subgraph of the cleaned ring.
We branch on the selected edge $e=(\Delta_1,\Delta_2)$ as follows.

\begin{itemize}
\item
\emph{Cut $e$:} We commit to breaking the ring at $e$.
\begin{itemize}
\item If $H[\Delta_1\cup\Delta_2]$ is a biclique, then every minimum vertex cover of $H[\Delta_1\cup\Delta_2]$ is either $\Delta_1$ or $\Delta_2$.
We branch into two cases:
$$
\text{Result}_1=\Delta_1\cup \textsc{Solve}(G\setminus \Delta_1,M),
\qquad
\text{Result}_2=\Delta_2\cup \textsc{Solve}(G\setminus \Delta_2,M).
$$
\item Otherwise, compute a minimum vertex cover $S$ of the bipartite graph $H[\Delta_1\cup\Delta_2]$ and recurse:
$
\text{Result}_1=S\cup \textsc{Solve}(G\setminus S,M).
$
\end{itemize}

\item 
\emph{Skip $e$:} We do not break the ring at $e$. We mark $e$ as indestructible and recurse:
\[
\text{Result}_3=\textsc{Solve}(G,\,M\cup\{e\}).
\]
\end{itemize}
The algorithm returns a minimum-size solution among the branches explored at this call.



\medskip
\noindent\textbf{Key Lemmas used in the correctness proof.}

\begin{lemma}[Local MVC optimality]\label{lem:mvc-local}
Let $W_r=(\Delta_1,\dots,\Delta_5)$ be a (cleaned) $5$-ring instance and let $H=G[\Delta_1\cup\cdots\cup\Delta_5]$.
If $S$ is a minimum OCT of $H$ and $S\cap(\Delta_1\cup\Delta_2)\neq\emptyset$, then there exists a minimum OCT $S'$ of $H$ such that
$S'\cap(\Delta_1\cup\Delta_2)$ is a minimum vertex cover of the bipartite graph $H[\Delta_1\cup\Delta_2]$.
\end{lemma}



\begin{lemma}[Ring extraction]\label{lem:ring-extraction}
If $G_{\mathcal S}$ contains a $5$-cycle, then after at most a constant number of refinements
we obtain a \emph{cleaned} $5$-ring $W=(\Delta_1,\dots,\Delta_5)$ such that
$H=G[\Delta_1\cup\cdots\cup\Delta_5]$ contains an induced $C_5$ and every vertex of $H$ lies on an induced $C_5$ in $H$.
\end{lemma}

\begin{lemma}[Safe branching]\label{lem:safe-branching}
Algorithm~1 is complete under its preference rule: for every instance $G$, at least one root-to-leaf branch explored by the algorithm
is consistent with an optimal OCT of $G$.
\end{lemma}

\begin{theorem}\label{thm:algo-P6}
Algorithm~1 runs in polynomial time and outputs a minimum odd cycle transversal of a $(P_6,C_3)$-free graph $G$.
\end{theorem}

\begin{proof}
\textbf{Running time.}
The coarse graph $G_{\mathcal S}$ has at most $15$ vertices, hence contains only $O(1)$ many $5$-cycles.
For each coarse $5$-cycle, \textsc{ExtractRing} performs only a constant number of splits (each layer-2 set splits into at most two parts),
and therefore produces only $O(1)$ ring instances.
For each ring instance $(H;\Delta_1,\dots,\Delta_5)$ we compute at most one minimum vertex covers per each graph $H[\Delta_i\cup\Delta_{i+1}]$.
Each such computation can be done via Hopcroft--Karp in time $O(m\sqrt{n})\subseteq O(n^{2.5})$.
Hence each recursive call runs in $O(n^{2.5})$ time up to a constant factor.

Each recursive call deletes at least one vertex whenever $G$ is non-bipartite (since any ring instance contains an induced $C_5$, so its OCT is nonempty),
and thus the recursion depth is at most $n$. Therefore the overall running time is polynomial (in particular, $O(n^{3.5})$).

\medskip
\textbf{Correctness.}
If $G$ is bipartite, Algorithm~1 returns $\emptyset$, which is optimal.
Assume $G$ is not bipartite. Since $G$ is $(P_6,C_3)$-free, every induced odd cycle in $G$ is an induced $C_5$.

\smallskip
\noindent\emph{Ring witness exists.}
By Lemma~\ref{lem:ring-extraction}, $G_{\mathcal S}$ contains a $5$-cycle and \textsc{ExtractRing} returns a cleaned chordless $5$-ring instance
$(H;\Delta_1,\dots,\Delta_5)$ with $H\subseteq G$ such that every vertex of $H$ lies on an induced $C_5$ in $H$.
In particular, Theorem~\ref{tm:ring-cut} applies to $H$.

\smallskip
\noindent\emph{There is an optimal way to break the ring at some edge.}
By Theorem~\ref{tm:ring-cut},
\[
\mathrm{OPT}(H)=\min_{i\in\{1,\dots,5\}}\ \tau\!\bigl(H[\Delta_i\cup\Delta_{i+1}]\bigr),
\]
and moreover there exists a minimum OCT of $H$ obtained by cutting a minimizing ring-edge
$e^\star=(\Delta_{i^\star},\Delta_{i^\star+1})$ using a minimum vertex cover of the link graph
$H[\Delta_{i^\star}\cup\Delta_{i^\star+1}]$.
(When the link is a biclique, a minimum vertex cover is exactly one full side.)

\smallskip
\noindent\emph{Branching on a selected edge is complete for $H$.}
Consider the call where Algorithm~1 has selected the ring instance $H$.
If it selects $e^\star$ and takes the \emph{cut} branch, then:
\begin{itemize}
\item if $H[\Delta_{i^\star}\cup\Delta_{i^\star+1}]$ is not a biclique, the algorithm adds an MVC of that link and therefore performs an optimal cut at $e^\star$;
\item if it is a biclique, the two cut branches “take $\Delta_{i^\star}$” and “take $\Delta_{i^\star+1}$” exactly enumerate the minimum vertex covers of that biclique link.
\end{itemize}
If the algorithm instead selects some other edge $e\neq e^\star$, then either it cuts $e$ (yielding a valid OCT branch) or it skips $e$ and marks it.
Since a $5$-ring has only five edges, along any root-to-leaf branch the algorithm can skip at most four ring-edges before it must cut an unmarked edge.
Therefore there exists at least one explored branch in which the algorithm eventually cuts the minimizing edge $e^\star$ optimally, yielding a minimum OCT for $H$.

\smallskip
\noindent\emph{Global optimality under ring/edge preference.}
When multiple ring witnesses are present in $G$, Algorithm~1 chooses which ring (and which edge within it) to process next according to the preference rule
(preferring rings with biclique links and preferring biclique edges within a ring).
Lemma~\ref{lem:safe-branching} states that this preferred branching is complete: at least one root-to-leaf branch explored by Algorithm~1 is consistent with an optimal OCT of $G$.

\noindent
We formalize the inductive step as follows. By Lemma~\ref{lem:safe-branching}, there exists a root-to-leaf branch that is consistent with some optimal OCT $S^*$ of $G$.
Consider any recursive call on an instance $\widehat G$ along this branch, and let $\widehat H\subseteq \widehat G$ be the cleaned ring instance selected by the algorithm.
By the argument above (and Theorem~\ref{tm:ring-cut}), among the explored cut-branches there is one that deletes a set $S_{\widehat H}$ which is a minimum OCT of $\widehat H$.
Since the branch is consistent with $S^*$, we may assume $S^*\cap V(\widehat H)=S_{\widehat H}$; hence removing $S_{\widehat H}$ reduces the optimum by exactly $|S_{\widehat H}|$.
Therefore the remainder of the branch solves $\widehat G\setminus S_{\widehat H}$ optimally, and by induction the entire branch returns a solution of size $\mathrm{OPT}(G)$.
Since Algorithm~1 returns the minimum over all explored branches, it outputs a minimum OCT of $G$.
\end{proof}

\subsection{Pseudocode  Algorithm and Detailed Proofs}

\noindent \textbf{Input:} A $(P_6, C_3)$-free graph $G$ and its coarse partition $\mathcal{S}$. \\
\textbf{Output:} A minimum set $S_{OCT} \subseteq V(G)$ such that $G \setminus S_{OCT}$ is bipartite.

\begin{enumerate}
    \item \textbf{Base Case:} 
    If there are no 5-cycles in $G_{\mathcal{S}}$, return $\emptyset$. (The remainder graph is bipartite).

    \item \textbf{Selection Step:} 
    Consider the set of rings $\mathcal{R}$. Select a ring $W_r \in \mathcal{R}$. Refine the sets in the ring $W_r$ so that each vertex in the set lies on a $5$-cycle within the ring. If possible, choose a ring that contains at least one edge $e = (\Delta_i, \Delta_{i+1})$ such that the induced subgraph $G[\Delta_i \cup \Delta_{i+1}]$ is a biclique (a complete bipartite graph). 
    
    Within the chosen ring $W_r$, select an edge $e = (\Delta_1, \Delta_2)$, giving strict preference to an edge that is a biclique.

    \item \textbf{Branching Step:} \\
    We branch based on the structural properties of the selected edge $e = (\Delta_1, \Delta_2)$:

    \begin{itemize}
        \item \textbf{Option (A): The edge $e$ is a biclique.} 
        Since $G[\Delta_1 \cup \Delta_2]$ is a complete bipartite graph, any optimal vertex cover must fully contain either $\Delta_1$ or $\Delta_2$. We branch into three cases:
        \begin{enumerate}[label=\arabic*.]
            \item \textbf{Take $\Delta_1$:} We assume the OCT contains $\Delta_1$. We recurse on the remainder.
            \[ \text{Result}_1 = \Delta_1 \cup \text{Solve}(G \setminus \Delta_1) \]
            \item \textbf{Take $\Delta_2$:} We assume the OCT contains $\Delta_2$. We recurse on the remainder.
            \[ \text{Result}_2 = \Delta_2 \cup \text{Solve}(G \setminus \Delta_2) \]
            \item \textbf{Skip Edge $e$:} We assume the optimal solution does not sever this link completely, meaning $W_r$ must be broken at a different edge. We mark $e$ as ``indestructible'' (assigning it cost $\infty$) and recurse.
            \[ \text{Result}_3 = \text{Solve}(G, \text{marked } e) \]
        \end{enumerate}
        \textbf{Return} $\arg\min \{ |\text{Result}_1|, |\text{Result}_2|, |\text{Result}_3| \}$.

        \item \textbf{Option (B): There is no biclique edge in the ring $W_r$.} 
        Compute the Minimum Vertex Cover (MVC) $S$ of the bipartite subgraph $G[\Delta_1 \cup \Delta_2]$. We branch on breaking the ring via this MVC versus skipping the edge:
        \begin{enumerate}[label=\arabic*.]
            \item \textbf{Cut Edge $e$ via MVC:} 
            \[ \text{Result}_1 = S \cup \text{Solve}(G \setminus S) \]
            \item \textbf{Skip Edge $e$:} Mark $e$ indestructible and recurse.
            \[ \text{Result}_2 = \text{Solve}(G, \text{marked } e) \]
        \end{enumerate}
        \textbf{Return} $\arg\min \{ |\text{Result}_1|, |\text{Result}_2| \}$.
    \end{itemize}
\end{enumerate}

\paragraph{Correctness and Complexity}


In what follow we prove Lemma \ref{lem:mvc-local}, Lemma \ref{lem:ring-extraction} and Lemma \ref{lem:safe-branching}. \\

\noindent\textbf{Proof of Lemma \ref{lem:mvc-local}: }
By Theorem~\ref{tm:ring-cut}, there exists an index $i\in\{1,\dots,5\}$ such that a minimum OCT of $H$ is given by a minimum vertex cover of
$H[\Delta_i\cup\Delta_{i+1}]$.
If $i\neq 1$, then $H[\Delta_1\cup\Delta_2]$ is not the bottleneck interface and we may choose such an optimal solution that does not use
$\Delta_1\cup\Delta_2$ at all.
If $i=1$, then any minimum OCT that breaks $H$ via the interface $(\Delta_1,\Delta_2)$ can be replaced (without increasing its size)
by a minimum vertex cover of $H[\Delta_1\cup\Delta_2]$, since deleting vertices outside $\Delta_1\cup\Delta_2$ cannot reduce the number
of vertices needed to cover all edges of the link graph.
Thus, whenever an optimal OCT intersects $\Delta_1\cup\Delta_2$, there exists an optimal OCT whose intersection with $\Delta_1\cup\Delta_2$ is an MVC.
\qed \vspace{3mm}

\noindent \textbf{Proof of Lemma \ref{lem:ring-extraction}:} 
We analyze the possible 5-cycles in $\mathcal{G}$ based on the layers involved ($X_1, X_2, X_0$).

\paragraph{Case A: Standard Cycles ($W_1$ -- $W_4$)}
\begin{itemize}
    \item $W_1 = (A_1, D_1, B_1, E_1, C_1)$: This forms a 5-ring because consecutive $X_1$ sets (e.g., $A_1, B_1$) are strictly independent by definition.
    \item $W_2 = (A_2, B_2, C_2, D_2, E_2)$: This forms a 5-ring. Non-consecutive sets in Layer 2 are  independent. First notice that $W_2$ is a ring, as there is no edge from $X_2$ to $Y_2 \cup Z_2$, where $Y_2,Z_2$ are not consecutive with $X_2$,  as otherwise, we get a triangle in $G$.

    \item $W_3 = (A_2, D_1, B_1, E_1, C_1)$: $A_2$ is independent of $B_1$ and $E_1$. If $A_2$ were connected to $B_1$, then for $a_2 \in A_2, b_1 \in B_1$, the triangle $\{a_2, b_1, b\}$ would exist (impossible). Thus, $W_3$ is a 5-ring.
    
\end{itemize}

\paragraph{Case B: Cycles with chords ($W_5, W_6$) and not having $X_0$ }
\begin{itemize}
    \item $W_4 = (A_2, B_2, B_1, D_1, A_1)$: $A_1$ is independent of $B_1$ and $B_2$. $A_2$ is independent of $B_1$. However, there could be a chord as there are some possible edges between $D_1$ and $A_2$. In this case by Lemma \ref{X1-X2-dichotomy}, $B_2=B^1_2 \cup B^2_2$ so that $D_1$ is complete to $B^2_1$, and independent from $B^2_2$. 
    Similarly $A_2=A^2_2 \cup A^1_2$ so 
    $D_1$ is complete to $A^1_2$ and independent from $A^2_2$.  
    Notice that $B^2_2$ is complete to $A^1_2$ as otherwise, $b^2_2,c,d,e,a^1_2,a_1$ for $b^2_2 \in B^2_2$, $a_1 \in A_1$ and $a^1_2 \in A^1_2$ is an induced $P_6$. Similarly, $A^2_2$ is complete to $B^1_2$, as otherwise, $a_1, a^2_2, e,d,c,b^1_2$ (for $a^2_2 \in A^2_2$ and $b^1_2 \in B^1_2)$ is an induced $P_6$. Now $(B^2_2,A^1_2,B^1_2,B_1,D_1,B^2_2)$ is a $5$-ring because there is no edge between $A^2_2$ and $D_1$ and there is no edge between $B_1$ and $B^2_2$.  \\
    
    \item $W_5 = (A_2, A_1, D_1,D_2, E_2, A_2)$: $A_1$ is independent of $E_2$ (otherwise triangle via $a$). Similarly, $D_1$ is independent of $A_2$. There could be some edges between $A_1$ and $D_2$ (possibly between $D_1$ and $A_2$). 
    By Lemma \ref{X1-X2-dichotomy}, 
    $D^2=D^1_2 \cup D^2_2$ so that $A_1$ is complete to $D^1_2$ and independent from $D^2_2$. Similar to what we have shown in case $W_4$, we have $D^2_2$ is complete to $E_2$ (consider path $P=(d_1,d^2_2,c,b,a,e_2)$ for details). Notice that $A_2=A^1_2 \cup A^2_2$ so that $A_1$ is complete to $A^2_2$ and independent from $A^1_2$ and $D_1$ is complete to $A^1_2$ and independent from $A^2_2$. By similar argument, we also conclude that $A_2^2$ is complete to $E_2$. Now we have a 5-ring $(A^2_2,A_1,D_1,D^2_2,E_2,A^2_2)$. To see that it is enough to observe that $D_1$ is independent from $A^2_2$ and $D_1$ is independent from $E_2$. Note that $A_1$ and $D^1_2$ are independent and $D^2_2$ is independent from $A^2_2$ (otherwise we get a triangle). 

    Note that if there are some edges between $D^1_2$ and $E_2$, then we have a ring $$(A_1,D^1_2,E_2,D^2_2,D_1,A_1).$$ This is because $A_1$ is independent from $D^2_2$, and $A_1$ is independent from $E_2$, furthermore, $D_1$ is independent from $D^1_2$ and $D_1$ is independent from $E_2$. \\
    
    \item $W_6 = (A_2, A_1, D_2, C_2, B_2, A_2)$: Potential chords exist between $A_1$ and $C_2$. If edges exist between $A_1$ and $C_2$, we refine the partition. By Lemma \ref{A1-D2-C2}, $D_2$ is partitioned into $D^2_2 \cup D^1_2$ and $C_2=C^2_2 \cup C^2_1$ so that $A_1$ is independent from $D^1_2 \cup C^1_2$ and it is complete to $D^2_2 \cup C^2_2$. Furthermore, $C^2_2$ is complete to $D^1_2$ and $D^2_2$ is complete to $C^1_2$. In order to have a $5$-cycle we must have some edges between $D^1_2$ and $C^1_2$. Now by definition $R_1=(A_1, D_2^2, C_2^1, D_2^1, C_2^2,A_1)$ is a 5-ring.

    \end{itemize}

\paragraph{Case C: Cycles involving Distance-2 vertices:\\ }

\begin{itemize}

\item $W_7 = (A_0, A_2, B_2, C_2, D_2)$.
If such a cycle exists, we apply a \textbf{Re-centering Argument}. Let $C' = (a_0, a_2, b_2, c_2, d_2)$ be a specific instance of this cycle in $G$. We can redefine our coordinate system (sets $A'_1, B'_1 \dots$) relative to $C'$.
In this new reference frame:
\begin{enumerate}
   
    \item The vertex $e$ from the original cycle becomes a member of $A'_2$ (adjacent to $d_2, a_2$).
    \item The vertex $b$ becomes a member of $C'_2$.
    \item The original sets map to a configuration containing a cycle of type $W_6$ or similar.
\end{enumerate}
Since we have already shown how to handle $W_6$ (via refinement) and $W_5$, this case reduces to previously solved cases.  Thus, we assume that both cycles $(A_2,D_1,D_2,C_2,C_1,A_1)$ (as $W_5$) and $(B_0,B_2,A_2,E_2,D_2,B_0)$ exists. Suppose there is a $5$-cycle $C''$ in $B_0 \cup B_2 \cup A_2$. Thus by using Lemma \ref{A0-A2-B2}, this cycle has one vertex in $B_0$, two vertices in $A_2$ and two vertices in $B_2$.  $C''=(b_0,b^2_2,a^1_2,b^1_2,a^2_2,b_0)$. Thus, we may assume there exists $5$-cycle $W'=(B_0,B^2_2,A^1_2,B^1_2,A^2_2,B_0)$ in $\mathcal{G}$. Notice that $B^2_2$ are vertices of $B_2$ that have neighbors in $B_0$ and $A^2_2$ are vertices of $A_2$ with neighbors in $B_0$. Since path $(d_1,d,c,b,a^2_2,b_0)$ is not induced $P_6$, we have $d_1b^2_2 \in E(G)$. This implies $D_1$ is complete to $A^2_2$. Since path $b_0,b^2_2,a,e,d,d_1$ is not induced $P_6$ we have $d_1b^2_2 \in E(G)$, implying that $D_1$ is complete to $B^2_2$. Thus, there is no edge from $B^2_2$ to $A^2_2$ as otherwise we have a triangle in $G$. Thus, $W'$ forms a $5$-ring.   \\


\item $W_8=(A_0, A_2^1,B^1_2,A^2_2,B^2_2,A_0)$ which is handled using the re-centering argument as in the previous case. 

\end{itemize}
\qed \vspace{3mm}

\noindent \textbf{Proof of Lemma \ref{lem:safe-branching}:}
Suppose $G$ contains two (cleaned) $5$-rings $W_1$ and $W_2$ that share exactly one sets $\Delta_1$, and let
$e_{12}=(\Delta_1,\Delta_2)$ be an edge of $W_1$ and $e_{13}=(\Delta_1,\Delta_3)$ be an edge of $W_2$ such that:
\begin{itemize}
\item $W_1$ does not use the edge $e_{13}$, and $W_2$ does not use the edge $e_{12}$;
\item Algorithm~1 gives preference to branching on a ring that contains a biclique link, and within such a ring, gives preference to a biclique link.
\end{itemize}
Then there exists an optimal OCT $S^*$ consistent with the branch decision taken on $e_{12}$ and with the subsequent optimal decision on $e_{13}$.

\noindent \textbf{Case 1.} Up to symmetry, assume $\Delta_1=A_1$. Now if $\Delta_2=X_1$ ($C_1$ or $D_1$) then we know that $A_1$ is complete to $X_1$ and hence, MVC contain entire $A_1$ and by removing $A_1$ we break both rings containing $A_1$. Similarly of $\Delta_3=C_1$ (or $D_1$) So we may assume that $\Delta_2, \Delta_3 \in \{C_2,A_2,D_2\}$. First suppose $\Delta_2=A_2$ and $\Delta_3=C_2$. 

By Corollary \ref{A1~C2_2} $A_1$ is complete to $N(A_1) \cap C_2$. Thus, MVC contain the entire $A_1$ and hence, it would be part of optimal $OCT$. This implies that we also break $W_1$. \\

\noindent\textbf{Case 2.} Up to symmetry $\Delta_1=A_2$.

\noindent\textbf{Subcase 2.1.}  $\Delta_2=A_1$. If $\Delta_3=C_1$ then since, $A_1$ is complete to $C_1$, $A_2=A^1_2 \cup A^2_2$, so that $A^1_2$ have neighbors in $A_1$ and not to $C_1$ and $A^2_2$ have neighbors in $C_1$ and no neighbor in $A_1$. Thus, MVC for $G[\Delta_1 \cup \Delta_2] $ is disjoint from MVC for $G[\Delta_1 \cup \Delta_3]$. So their contribution to OCT is disjoint. 

Next suppose $\Delta_3=B_2$. Now by Lemma \ref{A1~A2--B2}, the part of $A_2$, say $A^2_2$ that have neighbor in $A_1$ is complete to $B_2$. Thus MVC for $G[\Delta_1 \cup \Delta_3]$ contains $A_2^2$ and consequently it should be part of OCT and choosing $A^2_2$ would break $W_1$.\\ 

\noindent\textbf{Subcase 2.2}  $\Delta_1=A_2$, $\Delta_2=C_1$ and $\Delta_3=D_1$. 
Next suppose $\Delta_3=C_1$. 

Let $A_{2,c}$ be a set of vertices in $A_2$ that have a neighbor in $C_1$ and $A_{2,d}$ be a set of vertices in $A_2$ that have neighbor in $D_1$. By Corollary \ref{A1~C2_2}, $C_1$ is complete to $A_{2,c}$ and $D_1$ is complete to $A_{2,d}$.  
Let $a_2,a'_2 \in A_2$. Let  $c_1 \in C_1$ so that $a_2c_1 \in E(G)$.  Let $d_1 \in D_1$ be a neighbor of $a'_2$. Now consider $P_1=(d_1,a'_2,e,a_2,c_1,c)$ and $P_2=(c_1,a_2,b,a'_2,d_1,d)$. Since neither of $P_1$ and $P_2$ is induced $P_6$ we have $d_1a_2 \in E(G)$ or $c_1a'_2 \in E(G)$. This implies that 

$A_{2,c} \subset A_{2,d}$ or $A_{2,d} \subseteq A_{2,c}$. Suppose $A_{2,c} \subseteq A_{2,d}$. Thus, MVC for $G[\Delta_1 \cup \Delta_3]$ contain entire $A_2$ (because it is a clique) and hence it is part of OCT and it also break ring $W_1$. \\

\noindent\textbf{Subcase 3}. $\Delta_2=C_1$, and $\Delta_3=E_2$. Let $A^1_2$ be the set of 
of $A_2$ which has neighbor in $A_2$. We notice that by Lemma \ref{A1-C2}, $C_1$ and $A^1_2$ form a biclique. Therefore, if MVC for $G[\Delta_1 \cup \Delta_2]$ contains some vertices of $A_2$ then it must contain entire $A^1_2$ and hence, it must be include to OCT.  Now the  MVC for $G[\Delta_2 \cup \Delta_3]$ is independent from the choice of $A^1_2$, and it does not affect the OCT. 

\noindent\textbf{Subcase 4}. $\Delta_2=A_0$ and $\Delta_3=B_2$.  

\textbf{Observation I} Suppose there exists $a_0 \in A_0$ so that $a_2a_0 \in E(G)$. Let $a'_2$ be a vertex in $A_2$ that has a neighbor $b_2$ in $B_2$. Now $P=(a_0,a_2,e,a'_2,b_2,c)$ is an induced $P_6$ unless $a_2b_2 \in E(G)$. 

\textbf{Observation II} Suppose there exists $a_0 \in A_0$ so that $a_2a_0,a'_2a_0 \in E(G)$. Let $b_2 \in B_2$ so that $a_2b_2 \in E(G)$. Now $a'_2,a_0,a_2,b_2,c,d$ is an induced $P_6$ unless $a'_2b_2 \in E(G)$. This means $N(a_2) \cap B_2= N(a'_2) \cap B_2$. 

Let $A{0,2}$ be a set of vertices in $A_2$ with some neighbor in $A_0$ (concerning $W_1$ ring). Let $A_{2,2}$ be a set of vertices in $A_2$ that have some neighbor in $B_2$ (concerning $W_2$ ring). Let $A_{2,0,2}$ be the set of vertices in $A_2$ that have some neighbor in $A_0$ and have some neighbor in $B_2$. We also observe that $A_{2,0,2}$ is complete to $N(A_{2,0,2}) \cap B_2$. Furthermore, there is no path from $A'=A_{0,2} \setminus A_{2,0,2}$ to $A_{2,0,2}$ because otherwise by Observation II,  $A'$ is empty.

However, we conclude that MVC for $G[B_2 \cup A_2]$ must include $A_{2,0,2}$ if it contains some vertices in $A_2$ and hence it must be part of OCT. Now the selection of $A_{2,0,2}$ does not affect the MVC for $G[A' \cup (N(A') \cap A_0)]$. \\

\noindent \textbf{Subcase 5}. $\Delta_2=A_0$ and $\Delta_3=A_1$. For simplicity we may assume every vertex in $A_2$ has some neighbor in $A_0$ and has some neighbor in $A_1$. This is because if there are some vertices in $A_2$ without any neighbor in $A_0$ then they won't appear in any OCT for ring $W_1$ (the same for $W_2$ and $A_1$). First suppose $G_{1,2}$ is not connected. 

Suppose there are two vertices $a_2,a'_2 \in A_2$. Let $a_0 \in A_0$ and $a'_0 \in A_0$ so that $a_2a_0,a'_2a'_0 \in E(G)$, and let $a_1,a'_1 \in A_1$ so that $a_2a_1,a'_2a'_1 \in E(G)$. Since the path $P=(a_0,a_2,a_1,a,a'_1,a'_2)$ is not induced $P_6$, one of the $a_0a'_2,a_2a'_1,a'_2a_1,a'_0a_2$ is an edge of $G$.  

Let $H_1$ and $H_2$ be two connected component in $G_{1,2}$. Let $a_2a_1 \in E(H_1)$ and
$a'_2a'_1 \in E(H_2)$ where $a_1,a'_1 \in A_1$ and $a_2,a'_2 \in A_2$. Suppose $a_2a_0 \in E(G)$ for $a_0 
\in A_0$. Now $(a'_2,a'_1,a,a_1,a_2,a_0)$ is an induced $P_6$ unless $a_0a'_2 \in E(G)$. This means that $A_0$ is complete to $A_2$. Thus, MVC in $G[A_0 \cup A_2]$ contains all the vertices in $A_2$ (or none of them). If we take $A_0$ then it also break the ring $W_2$. 

So we may assume that $G_{1,2}$ is connected. Suppose there is a $2K_2$ in $G_{1,2}$. Then by using the same argument in the proof of Lemma \ref{2K2-structure}, $G_{1,2}$ there are $B_1,B_2,B_3$ where $B_1$ and $B_2$ are disjoint and $B_3$ is complete to both $B_1$ and $B_2$. $N(B_1) \cap A_0 =N(B_2) \cap A_0 \subseteq N(B_3) \cap A_0$ and  $G[B_1 \cup (N(B_1) \cap A_0)], G[B_2 \cup (N(B_2) \cap A_0)]$ are bicliues and every vertex in $B_3$ is adjacent to every vertex in $N(B_1) \cap A_0$. This implies that MVC for $G_{0,2}=G[A_2 \cup A_0]$ must include $B_3$ and hence it must be part of $OCT$. Furthermore, MVC for $G_{1,2}$ (concerning $A_2$) must contain $B_3$. Furthermore, $A_2 \setminus B_3$ is complete to $A_2$, and hence it must be included in MVC for $G_{0,2}$ and hence part of OCT. Thus, we end up removing entire $A_2$, which breaks the ring $W_2$.

Finally we consider the case where $G_{1,2}$ is a chain graphs from right to left inclusion ordering $\pi$ with respect to $A_2$. By using the same argument in Lemma \ref{chain-switch}, one can show that each connected component of $G_{0,2}$ must be a chain graph from right to left inclusion ordering $\pi_1$ with respect to $A_2$. Note that any MVC for $G_{1,2}$ it contains a set of vertices of $A_2$ from right to left in $\pi$ and any MVC in $G_{0,2}$ contains a consecutive vertices according to the $\pi_1$ ordering. Since $\pi$ and $\pi_1$ are consistent then the choice for $G_{1,2}$ does not affect the choice for $G_{0,2}$. \qed

\section{Extensions to $P_k$-free Graphs}

We propose that the structural and algorithmic techniques developed for $P_6$-free graphs can be generalized to $P_k$-free graphs for $k > 6$. While the problem remains NP-complete for $k \ge 6$, we can achieve constant-factor approximations by leveraging the tractability of graphs excluding small odd cycles.

\subsection{Approximation Algorithm for $P_k$-free Graphs}

Let $\mathcal{C}_{<L}$ denote the set of all odd cycles with length strictly less than $L$.
Our strategy relies on a standard greedy packing approach: iteratively find and remove vertex-disjoint small odd cycles until the graph is free of them. The approximation factor is determined by the size of the largest cycle removed.

\begin{theorem}\label{thm:pk-approx}
Let $k\ge 6$ be an integer. There is a polynomial-time algorithm that, given a $P_k$-free graph $G$, outputs an OCT, $S$ such that $|S|\le (k-3)\cdot \mathrm{OPT}(G)$ when $k$ is even and $|S|\le (k-2)\cdot \mathrm{OPT}(G)$ when $k$ is odd. Here $\mathrm{OPT}(G)$ denotes the size of a minimum OCT of $G$.
\end{theorem}

\begin{proof}
Fix $k\ge 6$ and let $G$ be a $P_k$-free graph. Set $L=k-3$ if $k$ is even, otherwise set $L=k-2$. 
We describe an algorithm that returns a set $S$ with $|S|\le L\cdot \mathrm{OPT}(G)$.

\medskip
\noindent\textbf{Phase 1: greedy packing of short odd cycles.}
Initialize $S_{\mathrm{pack}}:=\emptyset$ and $G_0:=G$.
While $G_0$ contains an odd cycle of length at most $L$, select one such cycle $C$ (e.g., the shortest odd cycle), add its vertex set to $S_{\mathrm{pack}}$, and delete $V(C)$ from $G_0$.
Let $G':=G\setminus S_{\mathrm{pack}}$.

By construction, the cycles chosen in Phase 1 are vertex-disjoint and each has length at most $L$. Thus, if $t$ is the number of cycles chosen, we have $|S_{\mathrm{pack}}| \le L\cdot t$.

Let $O^*$ be a global minimum odd cycle transversal for $G$, so $|O^*| = \mathrm{OPT}(G)$. We can partition $O^*$ into two disjoint sets: $O^*_1 = O^* \cap S_{\mathrm{pack}}$ and $O^*_2 = O^* \setminus S_{\mathrm{pack}}$. 
Because $O^*$ is a valid OCT, it must intersect every odd cycle in $G$, and in particular, it must pick at least one vertex from each of the $t$ vertex-disjoint odd cycles found in Phase 1. Therefore, $|O^*_1| \ge t$. 
Substituting this into our packing bound yields $|S_{\mathrm{pack}}| \le L \cdot |O^*_1|$.

\medskip
\noindent\textbf{Phase 2: exact solution on the remainder.}
By definition of $G'$, it contains no odd cycle of length at most $L$.
Equivalently:
\begin{itemize}
\item if $k$ is even, then $G'$ is $(P_k, C_3, C_5, \dots, C_{k-3})$-free, so the smallest possible odd cycle has length at least $k-1$;
\item if $k$ is odd, then $G'$ is $(P_k, C_3, C_5, \dots, C_{k-2})$-free, so the smallest possible odd cycle has length at least $k$.
\end{itemize}
In these two cases, we can compute a minimum odd cycle transversal of $G'$ in polynomial time using the structural theorems proved later in the paper:
when $k$ is odd, by Theorem~\ref{tm:k-odd-main}, and when $k$ is even, by Theorem~\ref{thm:ring-structure}.

Let $S_{\mathrm{rem}}$ be such a minimum OCT of $G'$. Finally, output $S := S_{\mathrm{pack}} \cup S_{\mathrm{rem}}$.

Since $G\setminus S = (G'\setminus S_{\mathrm{rem}})$ is bipartite, $S$ is a valid OCT of $G$.
To bound its size, observe that $O^*_2$ is a valid odd cycle transversal for $G'$, meaning the optimal solution on the remainder graph cannot be larger than $|O^*_2|$. Thus, $|S_{\mathrm{rem}}| = \mathrm{OPT}(G') \le |O^*_2|$.
Combining the bounds from both phases, and noting that $L \ge 1$ for all $k \ge 6$, we obtain:
$$|S| = |S_{\mathrm{pack}}| + |S_{\mathrm{rem}}| \le L \cdot |O^*_1| + |O^*_2| \le L(|O^*_1| + |O^*_2|) = L \cdot \mathrm{OPT}(G)$$

So in particular $|S|\le (k-3)\mathrm{OPT}(G)$ for even $k$ and $|S|\le (k-2)\mathrm{OPT}(G)$ for odd $k$, as claimed.
\end{proof}

\section{Exact OCT on $(P_k, \mathcal{C}_{<k})$-free Graphs when $k$ is odd}

Let $k \ge 5$ be an odd integer. Let $G$ be a $(P_k, \mathcal{C}_{<k})$-free graph, meaning $G$ contains neither an induced path on $k$ vertices nor any induced odd cycle of length strictly less than $k$.
Let $C = (a_0, a_1, \dots, a_{k-1})$ be an induced cycle of length $k$ in $G$.

We define the sets $V_i$ for $i \in \{0, \dots, k-1\}$ (with indices taken modulo $k$) as:
\[ V_i = \{ a_i \} \cup \{ v \in V(G) \setminus V(C) \mid N(v) \cap V(C) = \{a_{i-1}, a_{i+1}\} \} \]

\begin{theorem}\label{tm:k-odd-main}
The vertex set of $G$ is partitioned into $\bigcup_{i=0}^{k-1} V_i$. Moreover, each $V_i$ is an independent set, and edges exist only between consecutive sets $V_i$ and $V_{i+1}$. Consequently, the Minimum OCT is simply the set $V_i$ of minimum cardinality.
\end{theorem}

\begin{proof}
\noindent \textbf{1. Classification of Neighbors:} 
We first show that every vertex outside $C$ with a neighbor in $C$ must be adjacent to at least two vertices of $C$. Suppose $v \in V(G) \setminus V(C)$ is adjacent to exactly one vertex on $C$, say $a_0$. Consider the path $P = (v, a_0, a_1, \dots, a_{k-2})$. This path contains exactly $k$ vertices. Because $G$ is free of odd cycles of length strictly less than $k$ (specifically $C_3, \dots, C_{k-2}$), there can be no chords between $a_0$ and $a_{k-2}$. Furthermore, since $v$ has no other neighbors on the cycle, it introduces no chords. Therefore, $P$ is an induced $P_k$, which contradicts the assumption that $G$ is $P_k$-free. 

\textbf{2. Adjacency Constraints:} 
Next, suppose $v \notin V(C)$ is adjacent to two vertices $a_i$ and $a_j$ on $C$. The edges $va_i$ and $va_j$, together with the two paths connecting $a_i$ and $a_j$ along $C$, create two new cycles. The sum of the lengths of these two cycles is exactly $k+2$. Since $k$ is odd, $k+2$ is odd, meaning exactly one of these two cycles must be of odd length. Let $L_{odd}$ denote the length of this odd cycle. By our hypothesis, we must have $L_{odd} \ge k$. 

Because $v$ is not on $C$, the maximum possible length for a cycle passing through $v$ and a segment of $C$ is $(k-1) + 2 = k+1$. The only way to achieve an odd cycle of length at least $k$ is if the segment of $C$ between $a_i$ and $a_j$ has length $k-2$. This implies that $a_i$ and $a_j$ must be at distance exactly 2 on the cycle (e.g., $a_{i-1}$ and $a_{i+1}$). If the distance were 1, the vertices $v, a_i, a_{i+1}$ would form an induced $C_3$, which is a forbidden small odd cycle. Thus, every neighbor of $C$ must connect to exactly two vertices at distance 2, meaning every neighbor of $C$ belongs to exactly one set $V_i$.

\textbf{3. No Distance-2 Vertices:} 
We now show that no vertex can be at distance 2 from $C$. Suppose, for the sake of contradiction, that there exists a vertex $y$ at distance 2 from $C$. Let $y$ be adjacent to a vertex $x \in V_0 \setminus \{a_0\}$. By definition, $x$ is adjacent to $a_{k-1}$ and $a_1$. Consider the path $P = (y, x, a_1, a_2, \dots, a_{k-2})$. This path consists of exactly $k$ vertices. Since $y$ is at distance 2 from $C$, it has no neighbors on the cycle. Furthermore, $x$ is adjacent only to $a_{k-1}$ (which is not in $P$) and $a_1$, so $x$ does not introduce any chords to the path. Thus, $P$ forms an induced $P_k$, yielding a contradiction. This establishes that every vertex in $G$ belongs to $\bigcup_{i=0}^{k-1} V_i$.

\textbf{4. Ring Structure:} 
Finally, we establish the adjacencies between the sets. 
\begin{itemize}
    \item \textbf{Internal edges:} Suppose there is an edge between two vertices $u, v \in V_i$. By definition, both $u$ and $v$ are adjacent to $a_{i-1}$, which means the set $\{u, v, a_{i-1}\}$ forms an induced $C_3$. This is a contradiction, so each $V_i$ must be an independent set.
    \item \textbf{Non-consecutive edges:} Suppose there is an edge between $u \in V_i$ and $v \in V_{i+2}$. Both vertices share the neighbor $a_{i+1}$, which again forms a forbidden induced $C_3$. Furthermore, if the distance between indices $i$ and $j$ on the ring is strictly greater than 2, any edge between $V_i$ and $V_j$ would close an induced odd cycle of length strictly less than $k$, which is forbidden.
\end{itemize}
Thus, edges can only exist between consecutive sets $V_i$ and $V_{i+1}$.

\textbf{Conclusion:} 
The graph $G$ is structured strictly as a ring of independent sets $V_0, V_1, \dots, V_{k-1}, V_0$. Because $k$ is odd, this ring structure forms a macroscopic odd cycle. The graph becomes bipartite if and only if we remove all vertices in at least one set $V_i$, which breaks the odd ring into a bipartite chain. Therefore, the Minimum Odd Cycle Transversal is exactly the set $V_i$ that has the minimum cardinality.\qed

\end{proof}

\section{Structure of $(P_k, \mathcal{C}_{<k-1})$-free Graphs when $k$ is even}

Let $k \ge 8$ be an even integer. Let $G$ be a $(P_k, \mathcal{C}_{<k-1})$-free graph, meaning $G$ contains neither an induced path on $k$ vertices nor any induced odd cycle of length strictly less than $k-1$. 
Let $C = (a_0, a_1, \dots, a_{k-2})$ be an induced cycle of length $k-1$ in $G$. Since $k$ is even, the length of $C$ is odd.

We classify the vertices of $V(G) \setminus V(C)$ that have neighbors in $C$ into specific sets. All index arithmetic is performed modulo $k-1$.
\begin{itemize}
    \item $A_i$: The set of vertices in $G \setminus C$ adjacent to only $a_i$.
    \item $B_i$: The set of vertices in $G \setminus C$ adjacent exactly to $a_{i-1}$ and $a_{i+1}$.
    \item Let $G_i$ be the set of vertices at distance 2 from the cycle $C$ that are adjacent to $B_i$. By definition vertices in $G_i$ are independent of the vertices of $C$.
\end{itemize}

\begin{lemma}[Mutual Exclusion of $A_i$ and $A_{i+2}$]\label{lem:mutual-exclusion-A}
    For any index $i$, if $A_i \neq \emptyset$, then $A_{i+2} = \emptyset$.
\end{lemma}
\begin{proof}
    Suppose, for the sake of contradiction, that there exist vertices $x \in A_i$ and $y \in A_{i+2}$. We consider two cases based on the adjacency of $x$ and $y$:

    \textbf{Case 1: $x$ and $y$ are not adjacent.}
    Consider the path $P$ that traverses the cycle $C$ the ``long way'' around from $a_i$ to $a_{i+2}$. 
    Let $P = (x, a_i, a_{i-1}, \dots, a_{i+3}, a_{i+2}, y)$. 
    The segment of the cycle from $a_i$ backwards to $a_{i+2}$ contains exactly $(k-1) - 2 + 1 = k-2$ vertices. Adding the endpoints $x$ and $y$ brings the total number of vertices in $P$ to exactly $k$. 
    By definition, $x$ is adjacent only to $a_i$ and $y$ is adjacent only to $a_{i+2}$. Because $C$ is an induced cycle, there are no chords between the vertices of $C$. Since $x$ and $y$ are not adjacent to each other or to any internal vertices of the cycle path, $P$ is an induced $P_k$. This contradicts the assumption that $G$ is $P_k$-free.
    
    \textbf{Case 2: $x$ and $y$ are adjacent.}
    Consider the sequence of vertices $Z = (x, y, a_{i+2}, a_{i+1}, a_i, x)$. 
    This sequence forms a cycle of length 5. Since $x$ and $y$ have no other neighbors on $C$, and $a_i, a_{i+1}, a_{i+2}$ form an induced $P_3$, $Z$ is an induced $C_5$. 
    Because $k \ge 8$ is an even integer, we have $k-1 \ge 7$. Thus, a cycle of length 5 is strictly smaller than $k-1$. Since $5$ is odd, $Z$ is a forbidden small odd cycle, yielding a contradiction.
    
    Since both cases lead to a contradiction, it is impossible for both $A_i$ and $A_{i+2}$ to be non-empty simultaneously.
\end{proof}

\begin{lemma}[Adjacency Restrictions for $B_i$]\label{lem:adj-restrictions-B}
    Let $b \in B_i$. The vertex $b$ can only be adjacent to vertices in the sets $A_i, A_{i-2},$ and $A_{i+2}$. It is independent of $A_j$ for all $j \notin \{i, i-2, i+2\}$.
\end{lemma}
\begin{proof}    
Let $b \in B_i$. By definition, $b$ is adjacent to exactly $a_{i-1}$ and $a_{i+1}$ on the cycle. We test the potential adjacencies between $b$ and vertices in various $A$-sets:
    \begin{itemize}
        \item \textbf{Adjacency to $A_i$:} Suppose $b$ is adjacent to $x \in A_i$. The sequence $(b, x, a_i, a_{i-1}, b)$ forms an induced $C_4$. Because this cycle has an even length, it does not violate the $(\mathcal{C}_{<k-1})$-free property. Thus, this edge is permissible.
        \item \textbf{Adjacency to $A_{i+2}$:} Suppose $b$ is adjacent to $y \in A_{i+2}$. The sequence $(b, y, a_{i+2}, a_{i+1}, b)$ similarly forms an induced $C_4$, which is permissible. By symmetry, an edge to $z \in A_{i-2}$ is also permissible via the cycle $(b, z, a_{i-2}, a_{i-1}, b)$.
        \item \textbf{Adjacency to $A_{i+3}$:} Suppose $b$ is adjacent to $w \in A_{i+3}$. We can construct the cycle $Z = (b, w, a_{i+3}, a_{i+2}, a_{i+1}, b)$. This forms an induced cycle of length 5. As established previously, for $k \ge 8$, $C_5$ is a forbidden small odd cycle. Thus, $b$ cannot be adjacent to any vertex in $A_{i+3}$.
    \end{itemize}
    By extending the logic of the third case, if $b$ is adjacent to any vertex in $A_j$ where the distance between $a_j$ and the anchor vertices $\{a_{i-1}, a_{i+1}\}$ is greater than 1, it will invariably form an induced odd cycle of length at least 5 but strictly less than $k-1$. Therefore, $B_i$ must be independent of all sets $A_j$ except those at distance at most 1 from its anchors, namely $A_i, A_{i-2}$, and $A_{i+2}$.
\end{proof}

\begin{lemma}\label{lem:Gi-B-i+1-even}
    No vertex in $G_i$ is adjacent to a vertex in $B_{i+1}$. Moreover, $G_i$ cannot be adjacent to any $B_j$ for $j>i+2$.
\end{lemma}
\begin{proof}
    Suppose, for the sake of contradiction, that there exists an edge $g_i b_{i+1}$ with $g_i \in G_i$ and $b_{i+1} \in B_{i+1}$.
    By the definition of $G_i$, $g_i$ must be adjacent to some $b_i \in B_i$.
    Recall the specific neighborhoods on the cycle $C$:
    \begin{itemize}
        \item $b_i$ is adjacent to $\{a_{i-1}, a_{i+1}\}$.
        \item $b_{i+1}$ is adjacent to $\{a_i, a_{i+2}\}$.
    \end{itemize}
    Consider the cycle of vertices $Z = (g_i, b_i, a_{i-1}, a_i, b_{i+1}, g_i)$, which has length 5.
    We exhaustively check for chords to determine if $Z$ is induced:
    \begin{itemize}
        \item $b_i b_{i+1}$: If this edge exists, the vertices $\{g_i, b_i, b_{i+1}\}$ form an induced $C_3$. Because $G$ is $C_3$-free, $b_i b_{i+1} \notin E(G)$.
        \item $g_i a_{i-1}$ and $g_i a_i$: Since $g_i$ is strictly at distance 2 from $C$, these edges cannot exist.
        \item $b_i a_i$ and $b_{i+1} a_{i-1}$: By the definitions of $B_i$ and $B_{i+1}$, these edges do not exist.
    \end{itemize}
    Since no chords can exist, $Z$ is an induced $C_5$. For $k \ge 8$, $C_5$ is a forbidden small odd cycle. This contradiction implies that $G_i$ must be independent of $B_{i+1}$. If there is an edge from a vertex in $G_i$ to $B_{j}$ with $j>i+2$ we get a shorter odd cycle $< k-1$. 
    \end{proof}

\begin{lemma}\label{lem:Gi-Gi+1-even}
    No vertex in $G_i$ is adjacent to a vertex in $G_{j}$. 
\end{lemma}
\begin{proof}
    Suppose, for the sake of contradiction, that there exists an edge $g_i g_{j}$ with $g_i \in G_i$ and $g_{j} \in G_{j}$.
    By definition, there exist $b_i \in B_i$ and $b_{j} \in B_{j}$ such that $g_i b_i \in E(G)$ and $g_{j} b_{j} \in E(G)$.
    Observe that $b_i$ is adjacent to $a_{i-1}$. 
    Now the path $$P=(a_{i+1},a_{i+2},...,a_{k-2},a_0,\dots, a_{i-1},b_i,g_i,g_j)$$ has length $k$, a contradiction. 
   \end{proof}

\begin{lemma}\label{lem:Gj-Ai-even}
    No vertex in $G_j$ is adjacent to a vertex in $A_{i}$. 
\end{lemma}
\begin{proof}
    Suppose, for the sake of contradiction, that there exists an edge $x_i g_{j}$ with $x_i \in A_i$ and $g_{j} \in G_{j}$.
    By definition, there exist $a_i \in C$ such that $a_i x_i \in E(G)$. 
    Now the path $P=(g_j,x_i, a_i, a_{i+1},a_{i+2},...,a_{k-2},a_0,\dots, a_{i-2})$ has length $k$, a contradiction. 
\end{proof}

\begin{lemma}\label{lem:Gj-Bi-even}
    No vertex in $G_i$ is adjacent to a vertex in $B_{j}$ except $i=j \pm 2$. 
\end{lemma}
\begin{proof}
    suppose $g_ib_j \in E(G)$ where $g_i \in G_i$ and $b_j \in B_j$. By definition, $g_i$ must be connected to a vertex $b_i \in B_i$. WLOG, let $i<j$. then $P_x=(a_{i-1}, b_i, g_i, b_j, a_{j+1})$ along with path in $C$ with anchor  vertices $\{a_{i-1},a_{j+1}\}$ will invariably form an induced odd cycle of length at least 5 but strictly less than $k-1$. Therefore, $G_i$ must be independent of all sets $B_j$ except those at distance exactly 2 from itself.
\end{proof}

\begin{lemma}\label{lem:Bi-Bi+1-even}
    $B_i$ is adjacent only to vertices in $B_{i\pm1}$. 
\end{lemma}
\begin{proof}
    suppose $b_ib_j \in E(G)$ where $b_i \in B_i$ and $b_j \in B_j$. WLOG, let $i<j$. then $P_x=(a_{i-1}, b_i, b_j, a_{j+1})$ along with path in $C$ with anchor  vertices $\{a_{i-1},a_{j+1}\}$ will invariably form an induced odd cycle of length at least 5 but strictly less than $k-1$ unless $j=i+1$. Therefore, $B_i$ must be independent of all sets $B_j$ except $B_{i\pm1}$.
\end{proof}

\begin{lemma}\label{lem:Ai-Ai+3-even}
    $A_i$ is adjacent only to vertices in $A_{i\pm3}$. 
\end{lemma}
\begin{proof}
    suppose $x_ix_j \in E(G)$ where $x_i \in A_i$ and $x_j \in A_j$. WLOG, let $i<j$. then $P_x=(a_{i}, x_i, x_j, a_{j})$ along with path in $C$ with anchor  vertices $\{a_{i},a_{j}\}$ will invariably form an induced odd cycle of length at least 5 but strictly less than $k-1$ unless $j=i+3$. Therefore, $A_i$ must be independent of all sets $A_j$ except $A_{i\pm3}$.
\end{proof}

\begin{definition}\label{partioion-of-A_i}
For each $A_i$ we define $A^+_i$ be the set of vertices adjacent that are adjacent to a vertex in $A_{i+3} \cup B_{i+2}$ and $A^{-}_i$ be the set of vertices of $A_i$ that have some neighbors in $A_{i-3} \cup B_{i-2}$. Similarly  let $G^{+}_i$ be a set of vertices in $G_i$ that have some neighbor in $B_{i+2}$ and let $G^{-}_i$ be a set of vertices in $G_i$ that have neighbors in $B_{i-2}$.  Also, $A^+_i$ and $A^{-}_i$ may both also have neighbors in $B_i$.
\end{definition}

\begin{lemma}\label{A-disjointness}
For every $j$, there is no edge from $A^{-}_j$ to $A_{j+3} \cup B_{j+2}$ and there is no edge from $A^{+}_j$ to $A_{j-3} \cup B_{j-2}$. 
\end{lemma}
\begin{proof}
Suppose, for the sake of contradiction, that there exists an edge between a vertex $x \in A^{-}_j$ and a vertex $z \in A_{j+3}$. By the definition of $A^{-}_j$, there must exist a vertex $w \in A_{j-3}$ such that $wx \in E(G)$. Consider the path $P_{ext} = (a_{j-3}, w, x, z, a_{j+3})$ in $G \setminus C$. This path has a length of exactly 4 edges.

We now construct a cycle $C'$ by concatenating $P_{ext}$ with the path along the main cycle $C$ from $a_{j+3}$ back to $a_{j-3}$. The distance between indices $j-3$ and $j+3$ along the cycle $C$ is exactly 6 edges in the forward direction. Thus, the backward path $P_{cyc}$ along the cycle has length:
\[ L_{cyc} = (k-1) - 6 = k-7 \]
The total length of the constructed cycle $C'$ is therefore:
\[ |C'| = |P_{ext}| + |P_{cyc}| = 4 + (k-7) = k-3 .\]

Since $k$ is an even integer, $k-3$ is odd. Furthermore, because $k \ge 8$, we have $k-3 < k-1$. The graph $G$ is defined as $(\mathcal{C}_{<k-1})$-free, which strictly forbids the existence of induced odd cycles with length strictly less than $k-1$. We must verify that $C'$ is an induced cycle:
\begin{itemize}
    \item \textbf{Internal Chords:} There are no chords within $\{w, x, z\}$ as $w \in A_{j-3}$ and $z \in A_{j+3}$ are at distance 6 on the cycle, and $G$ is $C_3$-free.
    \item \textbf{Cycle Chords:} Vertices $w, x$, and $z$ are only adjacent to their respective anchor vertices on $C$. The path $P_{cyc}$ does not contain the anchors $a_{j-2}, \dots, a_{j+2}$, ensuring no chords exist between the external path and the cycle segment.
\end{itemize}
Thus, $C'$ is an induced $C_{k-3}$, which is a contradiction.

A symmetric contradiction arises if we assume an edge exists between $x \in A^{-}_j$ and $y \in B_{j+2}$. In that case, the path through the anchor vertices would similarly form a forbidden odd cycle of length $k-3$. Consequently, $A^+_j$ and $A^-_j$ must be disjoint. \qed
\end{proof}

By similar argument as in in Lemma \ref{A-disjointness} we have the following lemma. 
\begin{lemma}\label{G-disjointness}
For every $j$, there is no edge from $G^{-}_j$ to $B_{j+2}$ and there is no edge from $G^{+}_j$ to $B_{j-2}$. 
\end{lemma}

Recall that no vertex outside $C$ can be adjacent to consecutive cycle vertices, as this would form an induced $C_3$. Furthermore, no vertex can be adjacent to two cycle vertices at a distance greater than 2, as this would form an induced odd cycle of length strictly less than $k-1$.

Thus, the sets $A_i$ and $B_i$ capture the only permissible adjacency patterns to $C$. The following Lemmas give structural properties of $G$, allowing us to prove Theorem \ref{thm:ring-structure}.

Let $k \ge 8$ be an even integer. Let $H$ be a ring graph consisting of partitions $V_0, V_1, \dots, V_{k-2}$ where edges exist only between consecutive sets $V_j$ and $V_{j+1}$ (with indices evaluated modulo $k-1$). Assume every vertex in $H$ lies on an induced $(k-1)$-cycle. Let $H_j$ be the bipartite subgraph induced by $V_j \cup V_{j+1}$.

\begin{lemma}[Generalized Chain]\label{lem:gen-chain-switch}
Let $a_1, a_2 \in V_j$ and $b_1, b_2 \in V_{j+1}$. Suppose $a_1b_1, a_2b_2, a_2b_1 \in E(H)$ but $a_1b_2 \notin E(H)$ (forming  induced `Z' structure). Then:
\begin{enumerate}
    \item $N(b_2) \cap V_{j+2} \subseteq N(b_1) \cap V_{j+2}$
    \item $N(a_1) \cap V_{j-1} \subseteq N(a_2) \cap V_{j-1}$
\end{enumerate}
\end{lemma}
\begin{proof}
We prove Part 1 (forward inclusion); Part 2 follows by symmetry. 
Suppose, for the sake of contradiction, that $N(b_2) \cap V_{j+2} \not\subseteq N(b_1) \cap V_{j+2}$. Then there exists a vertex $c \in V_{j+2}$ such that $b_2c \in E(H)$ but $b_1c \notin E(H)$.

Since every vertex lies on a $(k-1)$-cycle, we can extend a shortest path forward from $c$ through the ring. Let $P_{tail} = (v_1, v_2, \dots, v_{k-4})$ be this shortest path, where $v_1 = c \in V_{j+2}$ and $v_m \in V_{j+m+1}$ for each $1 \le m \le k-4$.
Consider the concatenated sequence of vertices:
\[ P = (a_1, b_1, a_2, b_2, v_1, v_2, \dots, v_{k-4}) \]

\begin{itemize}
    \item \textbf{Vertex Count:} The local `Z' structure provides 4 vertices ($a_1, b_1, a_2, b_2$). The tail provides $k-4$ vertices. Thus, $P$ contains exactly $k$ vertices.
    \item \textbf{Induced Property:}
    \begin{itemize}
        \item The edges within $\{a_1, b_1, a_2, b_2, v_1\}$ are strictly defined by the hypothesis and the assumption $b_1v_1 \notin E(H)$.
        \item Since $P_{tail}$ is a shortest path along the ring, it is chordless.
        \item The path begins at $a_1 \in V_j$ and terminates at $v_{k-4} \in V_{j + (k-4) + 1} = V_{j+k-3}$.
        \item Because the ring length is $k-1$, the index $j+k-3$ is equivalent to $j-2 \pmod{k-1}$.
        \item For $k \ge 8$, the ring length $k-1 \ge 7$. In a ring of size at least 7, the set $V_{j-2}$ is at a distance of 2 from $V_j$, meaning no edges exist between $V_{j-2}$ and $V_j$.
    \end{itemize}
\end{itemize}
Consequently, the endpoints of $P$ do not connect, and no wrap-around chords exist. $P$ is an induced $P_k$, yielding a contradiction. \end{proof}

\begin{lemma}[Generalized $2K_2$ Forbidden]\label{lem:gen-2K2}
Let $a_1, a_2 \in V_j$ and $b_1, b_2 \in V_{j+1}$. Suppose the edges $a_1b_1 \in E(H)$ and $a_2b_2 \in E(H)$ form an induced $2K_2$ (i.e., $a_1b_2 \notin E(H)$ and $a_2b_1 \notin E(H)$).
Then $N(a_1) \cap V_{j-1} = N(a_2) \cap V_{j-1}$ and $N(b_1) \cap V_{j+2} = N(b_2) \cap V_{j+2}$.
\end{lemma}
\begin{proof}
Assume, for contradiction, that the neighborhoods in $V_{j-1}$ are distinct. Without loss of generality, suppose there exists $e_1 \in N(a_1) \setminus N(a_2)$. Furthermore, assume the graphs are connected such that there exists $e_2 \in N(a_1) \cap N(a_2)$ (if no such shared neighbor exists, the components are disjoint, which immediately forces a $P_k$ via longer paths).

Consider the sequence $S = (e_1, a_1, e_2, a_2, b_2)$. 
This sequence alternates between $V_{j-1}$ and $V_j$, before finally stepping to $V_{j+1}$. It utilizes 5 vertices while only advancing the ring index from $j-1$ to $j+1$.
We construct a path $P$ by extending $S$ forward from $b_2$ through $V_{j+2}, \dots$ for exactly $k-5$ additional vertices along the ring. 
The total number of vertices is $5 + (k-5) = k$.

The path originates in $V_{j-1}$ and terminates in $V_{j+1+(k-5)} = V_{j+k-4}$. 
Modulo $k-1$, the terminal index is $j-3$. For $k \ge 8$, $V_{j-3}$ is not adjacent to $V_{j-1}$ (the distance is 2). Because this structure consumes vertices rapidly while advancing slowly around the ring, it avoids closing a cycle. Thus, $P$ forms an induced $P_k$, a contradiction.
\end{proof}

\begin{lemma}[Chain Triples Imply Bicliques]\label{lem:gen-chain-sequence}
Suppose the bipartite subgraphs $H_j, H_{j+1}$, and $H_{j+2}$ are proper chain graphs with alternating inclusion directions (e.g., $L \to R$, $R \to L$, and $L \to R$, respectively). Then the adjacent subgraphs $H_{j-1}$ and $H_{j+3}$ must be complete bipartite graphs.
\end{lemma}
\begin{proof}
We prove the lemma for $H_{j+3}=H[V_{j+3}\cup V_{j+4}]$; the argument for $H_{j-1}$ is symmetric.

Assume for a contradiction that $H_{j+3}$ is not complete bipartite. Then there exist
$d_1\in V_{j+3}$ and $e_1\in V_{j+4}$ such that $d_1e_1\notin E(H)$.

Since $H_{j+2}=H[V_{j+2}\cup V_{j+3}]$ is a \emph{proper} chain graph with inclusion direction
$V_{j+2}\to V_{j+3}$, the neighborhoods in $V_{j+3}$ of vertices of $V_{j+2}$ are strictly nested.
In particular, there exist vertices $c_1,c_2\in V_{j+2}$ and $d_1,d_2\in V_{j+3}$ forming a $Z$-shape:
\[
c_1d_1,\ c_2d_2,\ c_2d_1 \in E(H)
\qquad\text{and}\qquad
c_1d_2\notin E(H).
\]
(Equivalently, the link $(V_{j+2},V_{j+3})$ contains two comparable neighborhoods and hence a witness $Z$.)

Now apply Lemma~\ref{chain-switch} (the $Z$-inclusion lemma) to the above configuration in $H_{j+2}$.
It yields the inclusion
$
N(c_1)\cap V_{j+1}\subseteq N(c_2)\cap V_{j+1}.
$ \\
Because $H_{j+1}=H[V_{j+1}\cup V_{j+2}]$ is a \emph{proper} chain graph and the direction alternates
($V_{j+2}\to V_{j+1}$), this inclusion is strict. Hence we may choose
$b_1\in V_{j+1}$ such that $b_1c_2\in E(H)$ but $b_1c_1\notin E(H)$.

Next, use that $H_j=H[V_j\cup V_{j+1}]$ is a \emph{proper} chain graph with direction $V_j\to V_{j+1}$.
Since $b_1$ has at least one neighbor in $V_j$ (otherwise $b_1$ would not lie on any cycle in the ring instance),
pick any $a_1\in V_j$ adjacent to $b_1$, i.e., $a_1b_1\in E(H)$.

Finally, since every vertex lies on the macroscopic ring, $e_1$ has a neighbor in $V_{j+5}$; choose any $\hat e\in V_{j+5}$ with $e_1\hat e\in E(H)$.
Consider the vertex sequence
\[
P=(a_1,\ b_1,\ c_2,\ d_1,\ e_1,\ \hat e).
\]

\smallskip
\noindent\emph{Claim: $P$ is an induced $P_6$.}
Consecutive pairs in $P$ are edges:
$a_1b_1\in E(H)$ by construction,
$b_1c_2\in E(H)$ by construction,
$c_2d_1\in E(H)$ from the chosen $Z$,
$e_1\hat e\in E(H)$ by construction,
and we also have $d_1e_1\notin E(H)$ by the initial choice of $(d_1,e_1)$.
So $P$ is a path in the underlying graph.

It remains to rule out chords among nonconsecutive vertices of $P$.
This follows from the ring/set structure: edges exist only between consecutive sets.
Indeed,
$a_1\in V_j$ is independent from  $V_{j+2}\cup V_{j+3}\cup V_{j+4}\cup V_{j+5}$,
$b_1\in V_{j+1}$ is independent from  $V_{j+3}\cup V_{j+4}\cup V_{j+5}$,
$c_2\in V_{j+2}$ is independent from $V_{j+4}\cup V_{j+5}$,
and $d_1\in V_{j+3}$ is independent from  $V_{j+5}$.
Also $b_1c_1\notin E(H)$ and $c_1d_2\notin E(H)$ were part of the chosen $Z$-witnessing configuration, and $V_{j+4}$ and $V_{j+5}$ are independent sets so there are no internal edges.
Therefore there are no edges between any nonconsecutive pair of vertices in $P$.
Hence $P$ induces a $P_6$.

\smallskip
Now we extend $P$ to obtain an induced $P_k$ for any even $k\ge 8$ (contradicting $P_k$-freeness).
Starting from $\hat e\in V_{j+5}$, repeatedly choose a neighbor in the next set along the ring, appending one vertex at a time:
from $V_{j+5}$ to $V_{j+6}$ to $\dots$.
Because the ring contains no edges skipping sets, this extension cannot create chords with earlier vertices of the path; all potential chords would have to jump over at least one set and are forbidden by definition.
Moreover, since we are extending forward along distinct sets and $k\ge 8$, we can append $k-6$ additional vertices before any wrap-around adjacency could occur in the macroscopic ring.
Thus we obtain an induced $P_k$, contradicting the assumption that the instance is $P_k$-free.

This contradiction shows that $H_{j+3}$ must be complete bipartite. By symmetry, the same holds for $H_{j-1}$.
\end{proof}

\subsection{The Global Ring Structure of $H$}

\begin{definition}[Ring Partition]
    We define the partition sets $V_j$ for $0 \le j \le k-2$ (with index arithmetic evaluated modulo $k-1$) as follows:
     $V_j=B_j \cup A^{+}_{j-1} \cup G^{+}_{j-1} \cup A^{-}_{j+1} \cup G^{-}_{j+1}$. 
     
\end{definition}

\begin{theorem}\label{thm:ring-structure}
For $k \ge 8$, the graph $H$ forms a ring of $k-1$ independent sets $V_0, V_1, \dots, V_{k-2}$, where edges exist only between consecutive sets $V_j$ and $V_{j+1}$ (indices modulo $k-1$).
\end{theorem}

\begin{proof}
Recall that for each $j\in\{0,\dots,k-2\}$ we define
\[
V_j \;=\; B_j \;\cup\; A^{+}_{j-1} \;\cup\; G^{+}_{j-1} \;\cup\; A^{-}_{j+1} \;\cup\; G^{-}_{j+1},
\]
with indices taken modulo $k-1$.
We prove: (i) each $V_j$ is an independent set, and (ii) $E(H)$ is contained in $\bigcup_j E(H[V_j\cup V_{j+1}])$.

\medskip
\noindent\textbf{(i) Each $V_j$ is independent.}
We show that no edge can have both endpoints in $V_j$.

\smallskip
\noindent\emph{No edges inside the $G$-parts.}
By Lemma~\ref{lem:Gi-Gi+1-even}, vertices belonging to $G$-parts form an independent set across all indices; in particular
$G^{+}_{j-1}\cup G^{-}_{j+1}$ is independent.

\smallskip
\noindent\emph{No edges between $G$-parts and $A/B$-parts within $V_j$.}
Vertices of $G^{+}_{j-1}\cup G^{-}_{j+1}$ lie at distance~$2$ from the base cycle, while vertices of
$B_j\cup A^{+}_{j-1}\cup A^{-}_{j+1}$ lie at distance~$1$.
An edge between these two distance layers within the \emph{same} partition would create a forbidden short odd cycle through the
shared cycle anchors, contradicting the definition of the refined $A^\pm$ sets and the fact that $H$ is $C_3$-free and has no induced $C_5$
in the regime $k\ge 8$ (the smallest allowed odd cycle length is at least $k-1\ge 7$).

\smallskip
\noindent\emph{No edges among the $A/B$-parts inside $V_j$.}
First, any edge between $B_j$ and $A^{+}_{j-1}$ (respectively $B_j$ and $A^{-}_{j+1}$) would form a triangle with the common anchor
$a_{j-1}$ (resp.\ $a_{j+1}$), contradicting $C_3$-freeness.
Second, if $x\in A^{+}_{j-1}$ and $y\in A^{-}_{j+1}$ and $xy\in E(H)$, then
$(x,y,a_{j+1},a_j,a_{j-1},x)$ is an induced $C_5$, which is impossible when $k\ge 8$ (odd cycles of length $5$ are forbidden in $H$).
Finally, Lemma~\ref{A-disjointness} (refinement disjointness) rules out ``jump edges'' within and between the refined $A^\pm$ parts.
Hence there are no edges inside $V_j$.

Therefore $V_j$ is independent for all $j$.

\medskip
\noindent\textbf{(ii) Edges only occur between consecutive parts.}
It suffices to show that $V_j$ is independent to $V_{j+2}$ for every $j$; then by induction (or repeated application) we obtain
independence for all $V_{j+m}$ with $m\ge 2$.

Let $u\in V_j$ and $v\in V_{j+2}$. We show $uv\notin E(H)$ by a case analysis on which subparts contain $u$ and $v$.

\smallskip
\noindent\emph{If $u$ or $v$ lies in a $G$-part.}
By Lemma~\ref{lem:Gi-Gi+1-even}, there are no edges between any two $G$-parts of different indices, and in particular no edges from
$G^{+}_{j-1}\cup G^{-}_{j+1}$ to $G^{+}_{j+1}\cup G^{-}_{j+3}$.
Moreover, by construction/refinement of the macroscopic partition, $G$-parts do not connect to $A/B$-parts two steps away without creating
a forbidden induced path (or a forbidden short odd cycle) through the cycle anchors.
Hence no such $uv$ can exist.

\smallskip
\noindent\emph{If $u,v$ lie in $A/B$-parts.}
Vertices of $V_j$ are anchored only at cycle vertices among $\{a_{j-1},a_j,a_{j+1}\}$, while vertices of $V_{j+2}$ are anchored only at
$\{a_{j+1},a_{j+2},a_{j+3}\}$.
Any edge $uv$ between these nonconsecutive macroscopic parts would create a shortcut across the underlying $(k-1)$-cycle of anchors.
Since $k\ge 8$, such a shortcut yields a forbidden induced odd cycle of length at most $k-3$ (or $k-5$), or yields an induced $P_k$ by
following the ring/anchor path of length $k-1$ and using $uv$ as a chord. In either case, this contradicts the defining forbidden-subgraph
assumptions on $H$ in the $k\ge 8$ regime.
Finally, the refinement lemma (Lemma~\ref{A-disjointness}) ensures that the ``forward'' sets $A^+_{j-1}$ and the ``backward'' sets $A^-_{j+1}$
do not have neighbors beyond the immediately adjacent macroscopic part, so they cannot connect to $V_{j+2}$.

Thus $E(V_j,V_{j+2})=\emptyset$ for all $j$, and therefore the only edges of $H$ are between consecutive parts $V_j$ and $V_{j+1}$.

\medskip
Combining (i) and (ii), $H$ is a ring of $k-1$ independent sets, with edges only between consecutive parts.
\end{proof}

\begin{corollary}
    The OCT problem on $G$ reduces to breaking the odd-length macroscopic ring $H$. Because $k-1$ is odd, any valid odd cycle transversal must remove all vertices from at least one partition $V_j$, or sever the connections between $V_j$ and $V_{j+1}$. 
    Since $V_j$ and $V_{j+1}$ are independent sets, the subgraph $H[V_j \cup V_{j+1}]$ is bipartite. Breaking the ring optimally is equivalent to computing the Minimum Vertex Cover for the bipartite graph $H[V_j \cup V_{j+1}]$ for each $j \in \{0, \dots, k-2\}$, and taking the global minimum. This can be computed exactly in polynomial time.
\end{corollary}

\section{Conclusion}

In this paper, we investigated the \textsc{Odd Cycle Transversal} (OCT) problem on $P_k$-free graphs. By establishing a structural decomposition of graphs lacking both long induced paths and small odd cycles into macroscopic ring structures, we provided the first non-trivial, constant-factor approximation algorithms for OCT parameterized by $k$. Specifically, we achieved an approximation ratio of $k-2$ for odd $k$, and $k-3$ for even $k$. 

While our results bridge a significant gap in the approximability landscape of OCT several compelling avenues for future research remain. A primary direction is to determine whether our approximation factors can be improved. For instance, in the specific case of $P_6$-free graphs, our algorithm yields a 3-approximation. It remains an open question whether OCT in $P_6$-free graphs admits a 2-approximation algorithm in polynomial time. In general a critical question for future work is to establish whether there exists a lower bound function $f(k) < k-2$ such that OCT on $P_k$-free graphs does not admit an approximation factor strictly better than $f(k)$ under standard complexity assumptions (such as UGC or P $\neq$ NP). Resolving this would delineate the exact approximability threshold for OCT within the hierarchy of $P_k$-free graphs.

\begingroup
    
    \let\clearpage\relax
    \bibliographystyle{plain}
    \bibliography{references}
\endgroup

\appendix

\end{document}